\documentclass[11pt, titlepage, reqno]{amsart}
\usepackage[utf8]{inputenc}
\usepackage[T1]{fontenc}
\usepackage[british]{babel}
\usepackage{times}
\usepackage{euler}
\usepackage{amsaddr, amsbsy, amscd, amsfonts, amsmath, mathrsfs, amssymb, amsthm}
\usepackage{dsfont, mathrsfs, mathtools}
\usepackage{hyperref}
\usepackage{enumitem}
\usepackage{xparse}
\usepackage[bottom, hang, flushmargin]{footmisc}

\usepackage[a4paper, hmargin=2cm, vmargin={1.5cm,2.5cm}, nomarginpar, hcentering, includeheadfoot]{geometry}
\newtheorem{theorem}{Theorem}[section]
\newtheorem{lemma}[theorem]{Lemma}

\theoremstyle{definition}

\newtheorem{remark}[theorem]{Remark} 
\numberwithin{equation}{section}

\makeatletter
\renewcommand\section{%
  \@startsection{section}%
    {1}
    {0em}
    {1.5cm \@plus 0.1ex \@minus -0.05ex}
    {0.75cm \@plus 0.2em}
    {\centering \large \scshape}
  }

  \renewcommand\subsection{%
  \@startsection{subsection}%
    {2}
    {0em}
    {0.75cm \@plus 0.1ex \@minus -0.05ex}
    {0.5cm \@plus 0.2em}
    {\scshape}
  }

\makeatother

\usepackage{soul,xcolor}
\usepackage{color}
\setstcolor{red}

\providecommand{\comment}[1]{}

\RenewDocumentCommand \newline {o}{
\hfill\\[\IfValueTF{#1}{#1}{0} pt]%
}
\newcommand{\nquad}{\mspace{-18mu}}

\newcommand{\n}{\noindent}

\newcommand{\be}{\begin{equation}}
\newcommand{\ee}{\end{equation}}

\newcommand{\Id}{\mathds{1}}
\providecommand{\abs}[1]{\lvert#1\rvert}
\renewcommand*{\vec}[1]{\boldsymbol{#1}}

\ProvideDocumentCommand \conjugate {s m}{
    \IfBooleanTF{#1}{%
        \overline{#2}%
    }{%
        \bar{#2}%
    }%
}

\renewcommand\widehat[1]{%
\savestack{\tmpbox}{\stretchto{%
    \scaleto{%
        \scalerel*[\widthof{\ensuremath{#1}}]{\kern.1pt\mathchar"0362\kern.1pt}%
        {\rule{0ex}{\textheight}}%
    }{\textheight}%
}{2.4ex}}%
\stackon[-6.9pt]{#1}{\tmpbox}%
}
\let\oldforall\forall
\renewcommand{\forall}{\oldforall\,}
\RenewDocumentCommand \to {o}{
    \IfValueTF{#1}{%
        \xrightarrow[\,#1\,]{\;}%
    }{%
        \,\rightarrow\,%
    }%
}

\newcommand{\N}{\mathbb{N}}
\NewDocumentCommand \Z {o}{%
    \IfValueTF{#1}{%
        \mathbb{Z}_{\,#1}%
    }{%
        \mathbb{Z}%
    }%
}

\thanks{The authors acknowledge the support of:
MUR grant Dipartimento di Eccellenza 2023-2027 of Dipartimento di Matematica, Politecnico di Milano, and Gran Sasso Science Institute, and Dipartimento di Matematica ``G. Castelnuovo'', ``Sapienza'' Università di Roma;
GNFM Gruppo Nazionale per la Fisica Matematica - INdAM;
the European Research Council (ERC) under the European Union’s Horizon 2020 research and innovation programme (ERC CoG UniCoSM, grant agreement n. 724939).}

\title[Rigorous derivation of the Efimov effect in a simple model]{Rigorous derivation of the Efimov effect in a simple model}

\author[D. Fermi]{Davide Fermi}
\address{\emph{Politecnico di Milano}, P.zza Leonardo da Vinci 32, 20133 (MI), Italy\\
and \emph{Istituto Nazionale di Fisica Nucleare}, Sezione di Milano, Italy\\
\textsf{davide.fermi@polimi.it}  }
\author[D. Ferretti]{Daniele Ferretti}
\address{\emph{Gran Sasso Science Institute}, Via Michele Iacobucci, 2 - 67100 (AQ), Italy\\ \textsf{daniele.ferretti@gssi.it}}
\author[A. Teta]{Alessandro Teta}
\address{\emph{Sapienza Universit\`a di Roma}, Piazzale Aldo Moro, 5 - 00185 (RM), Italy\\ \textsf{teta@mat.uniroma1.it}}
\begin{document}

\begin{abstract} 
We consider a system of three identical bosons in $\mathbb{R}^3$ with two-body zero-range interactions and a three-body hard-core repulsion of a given radius $a>0$. 
Using a quadratic form approach we prove that the corresponding Hamiltonian is self-adjoint and bounded from below for any value of $a$.
In particular this means that the hard-core repulsion is sufficient to prevent the fall to the center phenomenon found by Minlos and Faddeev in their seminal work on the three-body problem in 1961.  
Furthermore, in the case of infinite two-body scattering length, also known as unitary limit, we prove the Efimov effect, \emph{i.e.}, we show that the Hamiltonian has an infinite sequence of negative eigenvalues $E_n$ accumulating at zero  and fulfilling the asymptotic geometrical law  $\;E_{n+1} / E_n  \; \to \; e^{-\frac{2\pi}{s_0}}\,\; \,\text{for} \,\; n\to +\infty$ holds, where $s_0\approx 1.00624$. 
\newline[10]
\begin{footnotesize}
\end{footnotesize}
\end{abstract}
\maketitle

\section{Introduction}

The Efimov effect is an interesting physical phenomenon occurring in three-particle quantum systems in dimension three (\cite{efimov1,efimov2}, see also~\cite{NE}).
It consists in the appearing of an infinite sequence of negative eigenvalues $E_n$, with $E_n \to 0$ for $ n \to \infty$, of the three-body Hamiltonian if the two-particle subsystems do not have bound states and at least two of them exhibit a zero-energy resonance (or, equivalently, an infinite two-body scattering length). 
A remarkable feature of the effect is that the distribution of eigenvalues satisfies 
 the universal geometrical law
\begin{equation}\label{gela}
    \lim_{n\to+\infty}\frac{E_{n+1}}{E_n} \:=\: e^{-\frac{2\pi}{s}}\,,
\end{equation}
where the parameter $s>0$ depends only on the mass ratios and, possibly, on the statistics of the particles.

\n
According to an intuitive physical picture, the  three-particle bound states (or trimers) associated to the eigenvalues are determined by a long range, attractive effective interaction of kinetic origin, which is produced by the resonance condition and does not depend on the details of the two-body potentials.
Roughly speaking, in a trimer the attraction between two particles is mediated by the third one, which is moving back and forth between the two.
It should also be stressed that the Efimov effect disappears if the two-body potentials become more attractive causing the destruction of the zero-energy resonance.  
For interesting experimental evidence of Efimov quantum states see, \emph{e.g.}, \cite{kra}. 

\n
The first mathematical result on the Efimov effect was obtained by Yafaeev in 1974 \cite{yafaev}.
He studied a symmetrized form of the Faddeev equations for the bound states of the three-particle Hamiltonian and proved the existence of an infinite number of negative eigenvalues.
In 1993 Sobolev \cite{sobolev} used a slightly different symmetrization of the equations and proved the asymptotics
\begin{equation}\label{asob}
\lim_{z \to 0^-} \frac{N(z)}{ \abs{\log\abs{z}\mspace{0.75mu}}} = \frac{s}{2\pi}  \,,
\end{equation}
where $N(z)$ denotes the number of eigenvalues smaller than $z<0$.
Note that~\eqref{asob} is consistent with the law~\eqref{gela}.
In the same year Tamura \cite{tamura2} obtained the same result under more general conditions on the two-body potentials.
Other mathematical proofs of the effect were obtained by Ovchinnikov and Sigal in 1979 \cite{OS} and Tamura in 1991 \cite{tamura1} using a variational approach based on the Born-Oppenheimer approximation.
For more recent results on the subject, see ~\cite{BT} (for the case of two identical fermions and a different particle), \cite{gridnev1} (for a two-dimensional variant of the problem) and~\cite{gridnev2}.

We notice that in the above mentioned mathematical results a rigorous derivation of the law \eqref{gela} is lacking.

\n
It is also worth observing that, before the seminal works of Efimov, Minlos and Faddeev \cite{MF1, MF2} studied the problem of constructing the Hamiltonian for a system of three bosons with zero-range interactions in dimension three. It was known that such Hamiltonian cannot be defined considering only pairwise zero-range interactions. Minlos and Faddeev showed that a self-adjoint Hamiltonian can be constructed by imposing suitable two-body boundary conditions at the coincidence hyperplanes, \emph{i.e.}, when the positions of two particles coincide, and also a three-body boundary condition at the triple-coincidence point, when the positions of all the three particles coincide. 
They also proved that the Hamiltonian is unbounded form below, due to the presence of an infinite sequence of negative eigenvalues diverging to $-\infty$.
Such instability property can be seen as a fall to the center phenomenon and it is due to the fact that the interaction becomes too strong and attractive when the three particles are very close to each other.
A further interesting result of the analysis of Minlos and Faddeev, even if it is not explicitly emphasized, is the proof of the Efimov effect in the case of infinite two-body scattering length (corresponding to the resonant case), with a rigorous derivation of the law~\eqref{gela}. This in particular shows that the occurrence of the Efimov effect can be obtained also with zero-range interactions, the only crucial condition being the presence of an infinite two-body scattering length.
Such a result is somewhat tainted by the fact that the Hamiltonian in unbounded from below and therefore unsatisfactory from the physical point of view. 

\n
Our aim is to present a mathematical proof of the Efimov effect and law \eqref{gela} for a bounded from below Hamiltonian obtained by a slight modification of the Minlos and Faddeev Hamiltonian.
\n
We mention that the problem of constructing  a lower bounded Hamiltonian for a three-body system  with zero-range interactions has been recently approached in the literature (see, \emph{e.g.}, \cite{BCFT,FiTe,FeTe,miche}).
The idea is to introduce an effective three-body force acting only when the three particles are close to each other, preventing the fall to the center phenomenon.

\n
In  the present work we consider a Hamiltonian with two-body zero-range interactions and another type of three-body interaction.
More precisely, the effective three-body force is replaced by a three-body hard-core repulsion.
We shall prove that such Hamiltonian is self-adjoint and bounded from below and then prove the Efimov effect, \emph{i.e.}, the existence of an infinite sequence of negative eigenvalues satisfying~\eqref{gela} when the two-body scattering length is infinite.  

\n
Our work can be viewed as an attempt to make rigorous the original physical argument of Efimov. Indeed, Efimov takes into account three identical bosons and his approach is based on the replacement of the two-body potential with a boundary condition, which is essentially equivalent to consider a two-body zero-range interaction. Then, he introduces hyper-spherical coordinates and shows that if the two-body scattering length is infinite then the problem becomes separable and in the equation for the hyper-radius $R$ the long range, attractive effective  potential $- (s_0^2 +1/4)/R^2$ appears. The behavior for small $R$ of this potential is too singular and an extra boundary condition at short distance  must be imposed. After this ad hoc procedure, he obtains the infinite sequence of negative eigenvalues satisfying the law~\eqref{gela} as a consequence of the large $R$ behavior of the effective potential.

\n
The self-adjoint and bounded from below Hamiltonian constructed in this paper can be considered as the rigorous counterpart of the ad hoc regularization scheme mentioned above. Furthermore, we show that the eigenvalues and eigenvectors found in a formal way in the physical literature are in fact eigenvalues and eigenvectors of our Hamiltonian in a rigorous sense and, accordingly, we obtain a mathematical proof of~\eqref{gela}. 

Let us introduce some notation. Here and in the sequel: $\vec{x}_1, \vec{x}_2, \vec{x}_3 \in \mathbb{R}^3$ are the coordinates of the three bosons in a fixed inertial reference frame; the units of measure employed are such that $\hbar = m_1 = m_2 = m_3 = 1$. It is convenient to introduce the system of Jacobi coordinates $\vec{r}_{cm},\vec{x},\vec{y} \in \mathbb{R}^3$ defined as
\begin{equation*}
		\vec{r}_{cm} := {\vec{x}_1 + \vec{x}_2 + \vec{x}_3 \over 3}\,, \qquad 
		\vec{x} := \vec{x}_2 - \vec{x}_1\,, \qquad 
		\vec{y} := {2 \over \sqrt{3}} \left( \vec{x}_3 - {\vec{x}_1 + \vec{x}_2 \over 2} \right) .
	\end{equation*}
Correspondingly, we have $\vec{x}_1	 = \vec{r}_{cm} - {1 \over 2}\,\vec{x} - {1 \over 2 \sqrt{3}}\,\vec{y}$, $\vec{x}_2 = \vec{r}_{cm} + {1 \over 2}\,\vec{x} - {1 \over 2 \sqrt{3}}\,\vec{y}$, $\vec{x}_3	 = \vec{r}_{cm} + {1 \over \sqrt{3}}\,\vec{y}$. 
The transpositions $\sigma_{ij}$ ($i,j \in \{1,2,3\}$) exchanging the $i$\textsuperscript{th} and the $j$\textsuperscript{th} particles are represented by the following changes of coordinates
$\;\displaystyle{\sigma_{12} : (\vec{r}_{cm},\vec{x},\vec{y}) \to (\vec{r}_{cm},-\vec{x},\vec{y})\,,}$  
$\; \displaystyle{\sigma_{23} : (\vec{r}_{cm},\vec{x},\vec{y}) \to (\vec{r}_{cm},\tfrac{1}{ 2} \,\vec{x} + \tfrac{\sqrt{3}}{2}\,\vec{y},\tfrac{\sqrt{3}}{2}\,\vec{x} - \tfrac{1}{2}\,\vec{y}),} \; $ and 
$\displaystyle{\sigma_{31} : (\vec{r}_{cm},\vec{x},\vec{y}) \to (\vec{r}_{cm},\tfrac{1}{2}\,\vec{x} - \tfrac{\sqrt{3}}{ 2}\,\vec{y},- \tfrac{\sqrt{3}}{ 2}\,\vec{x} - \tfrac{1}{ 2}\,\vec{y}).}
$

\n
Upon factorizing the center of mass coordinate $\vec{r}_{cm}$ (\emph{i.e.}, adopting the center-of-mass reference frame) the heuristic Hamiltonian describing our three-boson system is expressed by 
\begin{equation}\label{eq:Hxy}
    H =  - \Delta_{\vec{x}} - \Delta_{\vec{y}} + V^{\mspace{1.5mu} hc}_a(\vec{x}, \vec{y})+ \delta(\vec{x}) + \delta\!\left(\tfrac{1}{2}\, \vec{x} - \tfrac{\sqrt{3}}{2}\, \vec{y}\right) + \delta\!\left(\tfrac{1}{2}\, \vec{x} + \tfrac{\sqrt{3}}{2}\, \vec{y}\right),
\end{equation}
where, at a formal level, $V^{\mspace{1.5mu} hc}_a$ indicates a hard-core potential corresponding to a Dirichlet boundary condition on the hyper-sphere of radius $a$ in $\mathbb{R}^6$, centered at $(\vec{x}, \vec{y}) = (\vec{0},\vec{0})$ and the ``$\delta$-potentials'' represent the zero-range interactions between the pair of particles $(1,2)$, $(2,3)$ and $(3,1)$ respectively.
Notice that
$$(\vec{x}_1 - \vec{x}_2)^2 + (\vec{x}_2 - \vec{x}_3)^2 + (\vec{x}_3 - \vec{x}_1)^2 = \tfrac{3}{2}\,(\abs{\vec{x}}^2 + \abs{\vec{y}}^2),$$
therefore the hard-core potential $V^{\mspace{1.5mu} hc}_a$ plays the role to prevent  the three bosons from reaching  the  triple-coincidence point $\vec{x}_1 = \vec{x}_2 = \vec{x}_3\mspace{1.5mu}$, avoiding the above mentioned fall to the center phenomenon.

\n The bosonic Hilbert space of states for our system is 
\begin{equation}\label{L2}
        L^2_{s}(\Omega_a) := \left\{ \psi \!\in\mspace{-1.5mu} L^2(\Omega_a) \left|\: \psi(\vec{x},\vec{y}) = \psi(-\vec{x},\vec{y}) = \psi\!\left(\tfrac{1}{2}\,\vec{x} + \tfrac{\sqrt{3}}{2}\,\vec{y}, \tfrac{\sqrt{3}}{2}\,\vec{x} - \tfrac{1}{2}\,\vec{y}\!\right) \right.\!\right\},
    \end{equation}
where
\begin{equation}\label{Omegaa}
    \Omega_a := \left\{(\vec{x},\vec{y}) \in \mathbb{R}^6 \,\left|\; \abs{\vec{x}}^2 + \abs{\vec{y}}^2 > a^2 \right.\right\}.   
\end{equation}
Definition~\eqref{L2} encodes the symmetry by exchange given by $\sigma_{12}$ and $\sigma_{23}$ which clearly imply also the condition corresponding to the exchange performed by $\sigma_{31}$, \emph{i.e.} $\psi(\vec{x},\vec{y}) = \psi\!\left(\frac{1}{2}\,\vec{x} - \frac{\sqrt{3}}{2}\,\vec{y}, - \frac{\sqrt{3}}{2}\,\vec{x} - \frac{1}{2}\,\vec{y}\!\right)$.
In the following we shall construct the rigorous counterpart of~\eqref{eq:Hxy} as a self-adjoint and bounded from below operator in $L^2_{s}(\Omega_a)$.  The first step is to interpret the formal unperturbed operator  $\;- \Delta_{\vec{x}} - \Delta_{\vec{y}} + V^{\mspace{1.5mu} hc}_a(\vec{x}, \vec{y})\;$ as  the Dirichlet Laplacian in $\Omega_a$, namely
    \begin{equation}\label{hd}
        \mbox{dom}\big(H_{D}\big) = L^2_{s}(\Omega_a) \cap H^1_0(\Omega_a) \cap H^2(\Omega_a)\,, \qquad H_{D} \psi = (-\Delta_{\vec{x}} - \Delta_{\vec{y}})\psi\,.  
    \end{equation}
It is well known that~\eqref{hd} is the self-adjoint and positive operator uniquely defined by  the positive quadratic form 
    \begin{equation}\label{eq:defQD}
        \mbox{dom}\big(Q_{D}\big) := L^2_{s}(\Omega_a) \cap H^1_0(\Omega_a)\,, \qquad
        Q_{D}[\psi] := \int_{\Omega_a}\hspace{-0.2cm} d\vec{x}\,d\vec{y}\, \Big( |\nabla_{\vec{x}} \psi|^2 + |\nabla_{\vec{y}} \psi|^2 \Big)\,.
    \end{equation}          
The second, and more relevant, step is to define a self-adjoint  perturbation of the Dirichlet Laplacian $H_D$  supported by the coincidence hyperplanes
\begin{equation*}
            \pi_{12} \mspace{-1.5mu}=\mspace{-1.5mu}\left\{(\vec{x},\vec{y}) \!\in\mspace{-1.5mu} \Omega_a \,\Big|\: \vec{x} \mspace{-1.5mu}= 0 \right\}\mspace{-1.5mu}, \mspace{30mu}
            \pi_{23}\mspace{-1.5mu}=\mspace{-1.5mu}\left\{ (\vec{x},\vec{y}) \!\in\mspace{-1.5mu} \Omega_a \,\Big|\: \vec{x} \mspace{-1.5mu}= \mspace{-1.5mu}\sqrt{3}\,\vec{y} \right\}\mspace{-1.5mu}, \mspace{30mu}
            \pi_{31}\mspace{-1.5mu}=\mspace{-1.5mu}\left\{ (\vec{x},\vec{y}) \!\in\mspace{-1.5mu} \Omega_a \,\Big|\: \vec{x}\mspace{-1.5mu} =\mspace{-1.5mu} -\mspace{1.5mu}\sqrt{3}\,\vec{y} \right\}\mspace{-1.5mu}.  
    \end{equation*}
Following the analogy with the one particle case~\cite{albeverio}, a natural attempt is to construct  an operator which, roughly speaking, acts   as the Dirichlet Laplacian $H_D$ outside the hyperplanes  and it is characterized by a (singular) boundary condition on each hyperplane. Specifically, given $\alpha \in \mathbb{R}$,  we demand that
    \begin{equation}\label{eq:deltabc}
        \psi(\vec{x},\vec{y}) = {\xi(\vec{y}) \over 4\pi\,\abs{\vec{x}}} + \alpha\,\xi(\vec{y}) + o (1)\,, \qquad \mbox{for fixed\, $\vec{y} \in B_{a}^{c}$\, and\, $\abs{\vec{x}} \to 0$}.
    \end{equation}
where $B_{a}^{c} := \mathbb{R}^3 \setminus B_a$ and $B_a = \big\{\vec{y} \!\in\! \mathbb{R}^3 \,\big|\, \abs{\vec{y}} \!<\! a \big\}$.
The above condition describes the interaction between particles 1 and 2. Due to the bosonic symmetry requirements, it also accounts for the interactions between the other two admissible pairs of particles. We recall that $-\alpha^{-1}$ has the physical meaning of two-body scattering length. 
The strategy for the mathematical proof is based on a quadratic form approach, similar to the one adopted  for the construction of singular perturbations of a given self-adjoint and positive operator in analogous contexts (see, \emph{e.g.}, \cite{teta, DFT, BCFT}). 

\n
More precisely, starting from the formal Hamiltonian~\eqref{eq:Hxy} characterized by the boundary condition~\eqref{eq:deltabc}, we construct the corresponding quadratic form $Q_{D,\alpha}$  and we formulate our main results (section 2).  

\n
The main technical part of the paper is the proof that  $Q_{D,\alpha}$ is closed and bounded from below  in $L^2_{s}(\Omega_a)$ (section 3).  

\n
Then, we define  the Hamiltonian $H_{D,\alpha}$ of our system as the unique self-adjoint and bounded from below operator associated to $Q_{D,\alpha}$ and we  give an explicit characterization of the domain and action of $H_{D,\alpha}$, providing also an expression for the associated resolvent operator (section 4).  

\n
Finally we show that for $\alpha=0$ the Efimov effect occurs and the law~\eqref{gela} holds (section 5).

\n
In Appendix \ref{app:g0} we collect some useful representation formulas for the integral kernel of the resolvent of the free Laplacian in $\mathbb{R}^6$ and for the integral kernel of the resolvent of the Dirichlet Laplacian in $\Omega_a$. 

\n
In Appendix \ref{bvp} we recall the solution of the eigenvalue problem in the case $\alpha=0$ following the treatment usually given in the physical literature.


\section{Formulation of the main results}

In this section we first give a heuristic argument to derive the quadratic form associated to the formal Hamiltonian~\eqref{eq:Hxy} and then we formulate our main results. 

\n
Let us introduce the potentials $G^{\lambda}_{ij} \xi_{ij}$ ($\lambda >0$)  produced by a suitable charge $\xi_{ij}$ with support concentrated on the hyperplane $\pi_{ij}$:
\begin{equation}\label{eq:GzijInt}
        \big(G^{\lambda}_{ij} \xi_{ij}\big)(\vec{X}) = \int_{\Omega_a}\!\!\! d\vec{X}'\; R^\lambda_{D}(\vec{X},\vec{X}')\;\xi_{ij}(\vec{X}')\,\delta_{\pi_{ij}}(\vec{X}')\,.
    \end{equation}
Here, for the sake of brevity, we have introduced the notation
    \begin{equation*}
        \vec{X} = (\vec{x},\vec{y})\,, \qquad \vec{X}' = (\vec{x}'\!,\vec{y}')\,,
\end{equation*}
and we have denoted by  $R^\lambda_{D}(\vec{X},\vec{X}')$  the integral kernel associated to the resolvent operator $R_{D}^{\lambda} := (H_D +\lambda)^{-1}$.

\n
For later convenience, we write the kernel $R_{D}^{\lambda}\big(\vec{X},\vec{X}'\big)$ as
    \begin{equation}\label{eq:RDG0gz}
        R_{D}^{\lambda}\big(\vec{X},\vec{X}'\big) = R_{0}^{\lambda}\big(\vec{X},\vec{X}'\big) + g^{\lambda}\big(\vec{X},\vec{X}'\big)   \,,
    \end{equation}
where  $R_{0}^{\lambda}\big(\vec{X},\vec{X}'\big)$ is the integral kernel associated to the resolvent operator of the free Laplacian in $\mathbb{R}^6$ and the  function $g^{\lambda}\big(\vec{X},\vec{X}'\big)$ is a reminder term, solving the following elliptic problem for any fixed $\vec{X}' \in \Omega_{a}$ and $\lambda>0$:
    \begin{equation}\label{eq:gzPDE}
        \left\{\begin{array}{ll} \!
            \displaystyle{(- \Delta_{\vec{X}} +\lambda) g^\lambda\big(\vec{X},\vec{X}'\big) = 0} & \displaystyle{\mbox{for\; $\vec{X} \in \Omega_{a}$}}\,, \vspace{0.15cm}\\
            \displaystyle{g^\lambda\big(\vec{X},\vec{X}'\big) = -\,R_{0}^{\lambda}\big(\vec{X},\vec{X}'\big)} & \displaystyle{\mbox{for\; $\vec{X} \in \partial\Omega_{a}$}}\,, \\
            \displaystyle{g^\lambda\big(\vec{X},\vec{X}'\big) \longrightarrow 0} & 
            \displaystyle{\mbox{for\; $\abs{\vec{X}} \!\to\! +\infty$}}\,.
        \end{array}\right.      
    \end{equation}
In Appendix~\ref{app:g0} we give explicit expressions for $R_{0}^{\lambda}\big(\vec{X},\vec{X}'\big)$, $g^{\lambda}\big(\vec{X},\vec{X}'\big)$  and then for $R_{D}^{\lambda}\big(\vec{X},\vec{X}'\big)$.

\n Furthermore,  the potentials $G^\lambda_{ij} \xi_{ij} $ fulfills the following equation in distributional sense
    \begin{equation}\label{eq:GzijEq}
        (H_D +\lambda) G^\lambda_{ij} \xi_{ij} = \xi_{ij}\, \delta_{\pi_{ij}}\,.
    \end{equation}

\n
By a slight abuse of notation, we set
    \begin{equation}\label{eq:Gz}
        G^\lambda \xi := G^\lambda_{12} \xi_{12} + G^\lambda_{23} \xi_{23} + G^\lambda_{31} \xi_{31}\,. 
    \end{equation}  
In order to ensure that $G^\lambda \xi$ actually meets the bosonic symmetries encoded in $L^2_{s}(\Omega_a)$, we must require
    \begin{equation}\label{eq:xisym}
        \xi_{12}(\vec{X}) = \xi(\vec{y})\,, \qquad\quad
        \xi_{23}(\vec{X}) = \xi_{31}(\vec{X}) = \xi(-2\vec{y})\,.
    \end{equation}
Taking this into account and noting that
    \begin{align*}
        \int_{\mathbb{R}^3}\hspace{-0.2cm} d\vec{y}'\, R^\lambda_{0}\big(\vec{x},\vec{y};\vec{x'}\!,\vec{y}'\big)
        &= {1 \over (2\pi)^6} \int_{\mathbb{R}^3}\hspace{-0.2cm} d\vec{y}' \int_{\mathbb{R}^3 \times \mathbb{R}^3} \hspace{-0.6cm} d\vec{h}\,d\vec{k}\; {e^{i \vec{h} \cdot (\vec{x} - \vec{x}') + i \vec{k} \cdot (\vec{y} - \vec{y}')} \over |\vec{h}|^2 + |\vec{k}|^2 +\lambda} \\
        &= {1 \over (2\pi)^3} \int_{\mathbb{R}^3}\!\!\!\! d\vec{h} \; {e^{i \vec{h} \cdot (\vec{x} - \vec{x}')} \over |\vec{h}|^2 +\lambda} 
        = {e^{-\sqrt{\lambda}\,|\vec{x} - \vec{x}'|} \over 4\pi\, |\vec{x} - \vec{x}'|}\;,
    \end{align*}
from Eq.~\eqref{eq:GzijInt} we infer the following:
    \begin{align*}
        & \big(G^\lambda_{12} \xi_{12}\big)(\vec{X}) 
            = \int_{B_{a}^{c}}\!\!\! d\vec{y}'\, R^\lambda_{D}\big(\vec{x},\vec{y};\vec{0},\vec{y}'\big)\;\xi(\vec{y}') \\
        & = \xi(\vec{y}) \int_{B_{a}^{c}}\!\!\! d\vec{y}'\, R^\lambda_{0}\big(\vec{x},\vec{y};\vec{0},\vec{y}'\big) 
            + \int_{B_{a}^{c}}\!\!\! d\vec{y}'\, R^\lambda_{0}\big(\vec{x},\vec{y};\vec{0},\vec{y}'\big)\,\big[ \xi(\vec{y}') - \xi(\vec{y}) \big] 
            + \int_{B_{a}^{c}}\!\!\! d\vec{y}'\, g^\lambda \big(\vec{x},\vec{y};\vec{0},\vec{y}'\big)\,\xi(\vec{y}') \\
        & = \xi(\vec{y})\,{e^{- \sqrt{\lambda}\,\abs{\vec{x}}} \over 4\pi \abs{\vec{x}}} - \xi(\vec{y}) \int_{B_{a}}\!\!\! d\vec{y}'\, R^\lambda_{0}\big(\vec{x},\vec{y};\vec{0},\vec{y}'\big) \\
            & \hspace{4cm} + \int_{B_{a}^{c}}\!\!\! d\vec{y}'\, R^\lambda_{0}\big(\vec{x},\vec{y};\vec{0},\vec{y}'\big)\,\big[ \xi(\vec{y}') - \xi(\vec{y}) \big] 
                + \int_{B_{a}^{c}}\!\!\! d\vec{y}'\, g^\lambda\big(\vec{x},\vec{y};\vec{0},\vec{y}'\big)\,\xi(\vec{y}')\,;
    \end{align*}
    \begin{gather*}
        \big(G^\lambda_{23} \xi_{23}\big)(\vec{X})
        = \int_{B_{a}^{c}}\!\! d\vec{y}'\, R^\lambda_{D}\!\left(\vec{x},\vec{y}; -\,\tfrac{\sqrt{3}}{2}\,\vec{y}'\!, -\,\tfrac{1}{2}\,\vec{y}'\right) \xi(\vec{y}')\,; \\
        \big(G^\lambda_{31} \xi_{31}\big)(\vec{X}) 
        = \int_{B_{a}^{c}}\!\! d\vec{y}'\, R^\lambda_{D}\!\left(\vec{x},\vec{y}; \tfrac{\sqrt{3}}{2}\,\vec{y}'\!, -\,\tfrac{1}{2}\,\vec{y}'\right) \xi(\vec{y}') \,. 
    \end{gather*}

\n
Notice that, due to the singularity for $\abs{\vec{x}}\to 0$,  the potential $G^\lambda_{12} \xi_{12}$ does not belong to $H^1(\Omega_a)$. Of course, the same is true for $G^\lambda_{23} \xi_{23}$ and $G^\lambda_{31} \xi_{31}$, due to the same kind of singularities for $|\vec{x} - \sqrt{3} \vec{y}| \to 0$ and  for $|\vec{x}+\sqrt{3} \vec{y}| \to 0$, respectively. 
Moreover, such a singular behavior is exactly of the same form appearing in~\eqref{eq:deltabc}. This fact suggests to write a generic element of the operator domain as
\[
    \psi = \varphi^{\lambda} + G^{\lambda} \xi\,, \qquad \mbox{with $\varphi^{\lambda} \in \mbox{dom}\big(H_{D}\big)$}\,.
\]
In view of the previous arguments, Eq.~\eqref{eq:deltabc} is equivalent to
    \begin{align}\label{eq:bc0}
        & \varphi^\lambda(\vec{0},\vec{y}) = \left(\! \alpha + \tfrac{\sqrt{\lambda}}{4\pi} \right)\! \xi(\vec{y}) 
          + \xi(\vec{y}) \int_{B_{a}}\!\!\! d\vec{y}'\, R^\lambda_{0}\big(\vec{0},\vec{y};\vec{0},\vec{y}'\big) - \int_{B_{a}^{c}} \!\!\! d\vec{y}'\, R^\lambda_{0}\big(\vec{0},\vec{y};\vec{0},\vec{y}'\big)\,\big[ \xi(\vec{y}') - \xi(\vec{y}) \big] \nonumber \\
        & - \int_{B_{a}^{c}}\!\!\! d\vec{y}' \,  g^\lambda \big(\vec{0},\vec{y};\vec{0},\vec{y}'\big)  \xi(\vec{y}')
            - \int_{B_{a}^{c}}\!\!\! d\vec{y}'\left[R^\lambda_{D}\!\left(\vec{0},\vec{y}; -\,\tfrac{\sqrt{3}}{2}\,\vec{y}'\!, -\,\tfrac{1}{2}\,\vec{y}'\right)\! + R^\lambda_{D}\!\left(\vec{0},\vec{y}; \tfrac{\sqrt{3}}{2}\,\vec{y}'\!, -\,\tfrac{1}{2}\,\vec{y}'\right) \right]\! \xi(\vec{y}') \,.
    \end{align}

\n We now proceed to compute the quadratic form associated to the formal Hamiltonian $H$ introduced in \eqref{eq:Hxy}.  
For functions of the form $\psi = \varphi^{\lambda} + G^{\lambda} \xi\mspace{1.5mu}$, taking into account that $(H_D + \lambda) G^{\lambda}\xi = 0$ outside the two-particle coincidence hyperplanes, a heuristic computation  yields
\begin{align}\label{qf1}
    \big\langle \psi\big| (H + \lambda) \psi \big\rangle &= \lim_{\varepsilon \to 0^{+}} \int_{\Omega_{a,\varepsilon} } \hspace{-0.4cm} d\vec{X}\; \conjugate*{\psi}\,\big(H_{D} + \lambda\big)\psi = \lim_{\varepsilon \to 0^{+}} \int_{\Omega_{a,\varepsilon} } \hspace{-0.4cm} d\vec{X}\; \conjugate*{(\varphi^{\lambda} + G^{\lambda} \xi)}\:\big(H_{D} + \lambda\big) (\varphi^{\lambda} + G^{\lambda} \xi) \nonumber\\
    &= \big\langle \varphi^{\lambda} \big| (H_{D} + \lambda) \varphi^{\lambda} \big\rangle + \lim_{\varepsilon \to 0^{+}} \int_{\Omega_{a,\varepsilon} } \hspace{-0.4cm} d\vec{X}\; \conjugate*{(G^{\lambda} \xi)}\:\big(H_{D} + \lambda\big) \varphi^{\lambda}\,,
\end{align}
where $\Omega_{a,\varepsilon}= \Omega_a \,\cap\, \{\abs{\vec{x}} > \varepsilon\} \,\cap\, \{\abs{\vec{x} - \sqrt{3}\vec{y}} > \varepsilon\} \,\cap\, \{\abs{\vec{x} + \sqrt{3}\vec{y}} > \varepsilon\}$. 
By means of Eqs.~\eqref{eq:GzijEq}, \eqref{eq:Gz}, \eqref{eq:xisym} and of the bosonic symmetry of $\varphi^{\lambda}$ one has
    \begin{align}\label{qf2}
        & \lim_{\varepsilon \to 0^{+}} \int_{\Omega_{a,\varepsilon} } \hspace{-0.5cm} d\vec{X}\; \conjugate*{(G^{\lambda} \xi)}\;\big(H_{D} + \lambda\big) \varphi^{\lambda} 
            = \int_{\Omega_a}\!\!\! d\vec{X}\, \Big( \delta_{\pi_{12}} \conjugate*{\xi_{12}} + \delta_{\pi_{23}} \conjugate*{\xi_{23}} + \delta_{\pi_{31}} \conjugate*{\xi_{31}} \Big)\varphi^{\lambda} 
            = 3 \int_{B^{c}_{a}}\!\!\! d\vec{y}\; \conjugate*{\xi(\vec{y})}\,\varphi^{\lambda}(\vec{0},\vec{y}) \,.
    \end{align}
Moreover, using the boundary condition~\eqref{eq:bc0} we find 
    \begin{align}\label{qf3}
        & \int_{B^{c}_{a}}\!\!\! d\vec{y}\; \conjugate*{\xi(\vec{y})}\;\varphi^{\lambda}(\vec{0},\vec{y}) 
            = \left(\alpha + \tfrac{\sqrt{\lambda}}{4\pi}\right)\! \|\xi\|_{L^2}^{2}
                +  \int_{B_{a}^{c} \times B_{a}} \!\!\! d\vec{y}\, d\vec{y}'\, R^{\lambda}_{0}\big(\vec{0},\vec{y};\vec{0},\vec{y}'\big)\, \abs{\xi(\vec{y})}^2  \nonumber \\
        &\quad -   \int_{B_{a}^{c} \times B_{a}^{c}} \mspace{-20mu} d\vec{y}\, d\vec{y}'\, R^{\lambda}_{0}\big(\vec{0},\vec{y};\vec{0},\vec{y}'\big)\,\conjugate*{\xi(\vec{y})}\big[ \xi(\vec{y}') - \xi(\vec{y}) \big]
            -  \int_{B_{a}^{c} \times B_{a}^{c}} \mspace{-20mu} d\vec{y}\, d\vec{y}'\, g^{\lambda}\big(\vec{0},\vec{y};\vec{0},\vec{y}'\big)\; \conjugate*{\xi(\vec{y})}\,\xi(\vec{y}') \nonumber \\
        &\quad - 2 \int_{B_{a}^{c} \times B_{a}^{c}} \mspace{-20mu} d\vec{y}\, d\vec{y}'\, R^{\lambda}_{D}\!\left(\vec{0},\vec{y}; \tfrac{\sqrt{3}}{2}\vec{y}'\!, -\tfrac{1}{2}\vec{y}'\mspace{-1.5mu}\right) \conjugate*{\xi(\vec{y})}\,\xi(\vec{y}')\,. 
    \end{align} 

\n
By~\eqref{qf1}, \eqref{qf2} and~\eqref{qf3} we obtain the expression for the quadratic form associated to $H$. In order to give a precise mathematical definition, we first consider the quadratic form in $L^2(B_{a}^{c})$ appearing in~\eqref{qf3}:
    \begin{gather}
        \Phi_{\alpha}^{\lambda}[\xi] := \left(\!\alpha + \tfrac{\sqrt{\lambda}}{4\pi}\right)\! \|\xi\|_{L^2}^{2}+ \Phi_{1}^{\lambda}[\xi] + \Phi_{2}^{\lambda}[\xi] + \Phi_{3}^{\lambda}[\xi] + \Phi_{4}^{\lambda}[\xi]\,, \label{phila} \\
        \mbox{dom}\big(\Phi_{\alpha}^\lambda\big) :=  H^{1/2}\big(B_{a}^{c}\big) \cap L^2_{w}\big(B_{a}^{c} \big)\,, \label{dophila}
    \end{gather}
with
    \begin{gather}
        \Phi_{1}^{\lambda}[\xi] := \int_{B_{a}^{c}}\hspace{-0.3cm} d\vec{y} \left(\int_{B_{a}}\hspace{-0.3cm} d\vec{y}'\, R^{\lambda}_{0}\big(\vec{0},\vec{y};\vec{0},\vec{y}'\big) \right) \abs{\xi(\vec{y})}^2 \,, \label{eq:Phi1def}\\
        \Phi_{2}^{\lambda}[\xi] := {1 \over 2}\int_{B_{a}^{c} \times B_{a}^{c}} \hspace{-0.6cm} d\vec{y}\, d\vec{y}'\, R^{\lambda}_{0}\big(\vec{0},\vec{y};\vec{0},\vec{y}'\big)\,\big| \xi(\vec{y}) - \xi(\vec{y}') \big|^2 \,, \label{eq:Phi2def}\\
        \Phi_{3}^{\lambda}[\xi] := - \int_{B_{a}^{c} \times B_{a}^{c}} \mspace{-24mu} d\vec{y}\, d\vec{y}'\, g^{\lambda}\big(\vec{0},\vec{y};\vec{0},\vec{y}'\big)\, \conjugate*{\xi(\vec{y})}\,\xi(\vec{y}')\,, \label{eq:Phi3def}\\
        \Phi_{4}^{\lambda}[\xi] := -\,2 \int_{B_{a}^{c} \times B_{a}^{c} } \mspace{-24mu} d\vec{y}\, d\vec{y}' \, R^{\lambda}_{D}\!\left(\vec{0},\vec{y}; \tfrac{\sqrt{3}}{2}\,\vec{y}'\!, -\,\tfrac{1}{2}\,\vec{y}'\right)\! \conjugate*{\xi(\vec{y})}\,\xi(\vec{y}')\,. \label{eq:Phi4def}
    \end{gather}
    
\n
In \eqref{dophila} we have introduced the weighted $L^2$ space
    \begin{equation*}
        L^2_{w}\big(B_{a}^{c}\big) := L^2\big(B_{a}^{c},\,w\,d\vec{x}\big)\,,
    \end{equation*}
where for any  $b > a$, the continuous function $w$ is given by 
    \begin{equation}\label{eq:wdef}
        w(\vec{x}) = {b-a \over \abs{\vec{x}} - a}\,\Id_{B_b \,\cap\, B_{a}^{c} }(\vec{x}) + \Id_{B_{b}^{c} }(\vec{x}) \;, \qquad \mbox{for\; $\vec{x} \in B_{a}^{c}$}\,,
    \end{equation}
and $H^{1/2}(B_a^c)$ denotes the Sobolev-Slobodecki\u\i~space of fractional order $1/2$ given by
\begin{equation}\label{eq:H12def}
        H^{1/2}\big(B_{a}^{c}\big) := \left\{ \xi \in L^2\big(B_{a}^{c} \big) \;\left|\; \|\xi\|_{H^{1/2}}^{2} := \|\xi\|_{L^2}^{2} + \int_{B_{a}^{c} \times B_{a}^{c}}\hspace{-0.5 cm} d\vec{y}\,d\vec{y}'\,{\big|\xi(\vec{y}) - \xi(\vec{y}')\big|^2 \over |\vec{y} - \vec{y}'|^4} < \infty \right.\right\} .
    \end{equation}
Notice that the choice of the parameter $b>a$ is irrelevant and, since $w \geqslant 1$, we have 
    \begin{equation}\label{eq:L2L2w}
        \|\xi\|_{L^2} \leqslant \|\xi\|_{L^2_{w}}\,, \qquad \mbox{for all $\xi \in L^2\big(B_{a}^{c}\big)$}\,.
    \end{equation}
We  remark that the reason for the choice of~\eqref{dophila}  as the form domain of $\Phi_{\alpha}^\lambda$ will be clear in the course of the proofs reported in section 3. 
We also point out that the form domain~\eqref{dophila} is isomorphic (as a Hilbert space) to the Lions-Magenes space $H^{1/2}_{00}\big(B_{a}^{c}\big)$ \cite[p. 66]{LM}.

We are now in position  to define the quadratic form in $L_{s}^2(\Omega_a)$ 
    \begin{gather}
        Q_{D,\alpha}[\psi] := Q_{D}\big[\varphi^{\lambda}\big] + \lambda \,\|\varphi^{\lambda}\|_{L^2}^2 - \lambda \,\|\psi\|_{L^2}^2  + 3\,\Phi_{\alpha}^{\lambda}[\xi]\,, \label{eq:Qa} \\
        \mbox{dom}\big(Q_{D,\alpha}\big) := \left\{\psi = \varphi^{\lambda}\! + G^{\lambda} \xi \,\left|\, \varphi^{\lambda} \!\in\! L^2_s(\Omega_a) \cap H^1_{0}(\Omega_a)\,,\;\, \xi \!\in\! \mbox{dom}\big(\Phi_{\alpha}^\lambda \big) \,,\;\, \lambda \!>\! 0 \right.\right\}\,, \label{eq:DomQa}
    \end{gather}
where $Q_{D}$ is the Dirichlet quadratic form defined in~\eqref{eq:defQD}. 

\n
We stress that the definition of the quadratic forms~\eqref{phila}, \eqref{dophila} and~\eqref{eq:Qa}, \eqref{eq:DomQa} is the starting point of our rigorous analysis. 

Before proceeding, let us mention an equivalent characterization of the potential $G^{\lambda} \xi$. This will play a role in some of the proofs presented in the following sections. Let $\tau_{ij} : \mbox{dom}\big(H_{D}\big) \subset H^2(\Omega_a) \to H^{1/2}(\pi_{ij})$ be the Sobolev trace operator defined as the unique bounded extension of the evaluation map $\tau_{ij} \varphi := \varphi \!\upharpoonright\! \pi_{ij}$ acting on smooth functions $\varphi \in C^{\infty}_c(\Omega_a)$. We set
    \begin{equation*}
        \vec{\tau} := \tau_{12} \oplus \tau_{23} \oplus \tau_{31} : \mbox{dom}\big(H_{D}\big) \to L^{2}(\pi_{12}) \oplus L^{2}(\pi_{23}) \oplus L^{2}(\pi_{31})\,.
    \end{equation*}
It is crucial to notice that the range of $\vec{\tau}$ actually keeps track of the bosonic symmetry encoded in $\mbox{dom}\big(H_{D}\big) \subset L^2_s(\Omega_a)$. Taking this into account, noting that $R^{\lambda}_{D} : L^2_{s}(\Omega_a) \to \mbox{dom}\big(H_D\big)$ and using the natural embedding $H^{1/2}(\pi_{ij}) \hookrightarrow L^{2}(\pi_{ij})$, it is easy to check that
    \begin{equation}\label{eq:GzHs}
        G^{\lambda} = \big(\vec{\tau}\, R^{\lambda}_{D} \big)^{*} \,:\, \mbox{dom}\big(G^{\lambda}\big) \subset L^{2}(\pi_{12}) \oplus L^{2}(\pi_{23}) \oplus L^{2}(\pi_{31}) \to L^2_{s}(\Omega_a)\,,
    \end{equation}
where, in compliance with~\eqref{eq:xisym} and~\eqref{dophila}, we put
    \begin{equation*}
        \mbox{dom}\big(G^{\lambda}\big) := \Big\{ \vec{\xi} = (\xi_{12}, \xi_{23}, \xi_{31}) \;\Big|\;
        \xi_{12}(\vec{y}) = \xi(\vec{y})\,,\; \xi_{23}(\vec{y}) = \xi_{31}(\vec{y}) = \xi(-2\vec{y})\,,\;\,
        \xi \!\in\! \mbox{dom}\big(\Phi_{\alpha}^\lambda\big) \Big\}\,.
    \end{equation*}
Accordingly, with a slight abuse of notation we can rephrase Eq.~\eqref{eq:Gz} as
    \begin{equation*}
        G^{\lambda} \vec{\xi} \equiv G^{\lambda} \xi\,, \qquad \mbox{for all\; $\vec{\xi} \in \mbox{dom}\big(G^{\lambda}\big)$}\,.
    \end{equation*}
In the sequel we shall refer especially to the bounded operator
    \begin{equation}\label{eq:deftau}
        \tau \equiv \tau_{12} : \mbox{dom}\big(H_{D}\big) \to L^2(\pi_{12}) \equiv L^2(B^{c}_{a}) \,.
    \end{equation}

In the rest of this section we formulate the main results of the paper.

\begin{theorem}[Closedness and lower-boundedness of $Q_{D,\mspace{1.5mu}\alpha}$]\label{thm:Qa} \newline 
    (i)  The quadratic form $\Phi_{\alpha}^{\lambda}$ in $L^2(B_{a}^{c})$ defined by~\eqref{phila}, \eqref{dophila} is closed and bounded from below for any $\lambda>0$. More precisely, there exists a constant $B > 0$ such that
    \begin{equation}\label{eq:upPhi}
            \big| \Phi_{\alpha}^{\lambda}[\xi] \big| \leqslant B \left(\sqrt{\lambda}\, \|\xi\|_{L^2}^2 + \|\xi\|^2_{L^2_{w}} + \|\xi\|_{H^{1/2}}^2 \right), \qquad \mbox{for any\, $\lambda > 0$}\, .
    \end{equation}  
Furthermore, there exist $\lambda_0 > 0$, $A_0 > 0$ and $A(\lambda) > 0$, with $A(\lambda) = o(1)$ for $\lambda \to +\infty$, such that
    \begin{equation}\label{eq:loPhi}
            \Phi_{\alpha}^{\lambda}[\xi] \geqslant A_{0} \sqrt{\lambda}\, \|\xi\|_{L^2}^2 + A(\lambda) \left( \|\xi\|^2_{L^2_{w}} + \|\xi\|^2_{H^{1/2}} \right), \qquad \mbox{for any\, $\lambda > \lambda_0$}\,.
    \end{equation}

\n
(ii) The quadratic form $Q_{D,\alpha}$ in $L^2_{s}(\Omega_a)$ defined by~\eqref{eq:Qa}, \eqref{eq:DomQa} is  independent of $\lambda$,  closed and lower-bounded.
\end{theorem}

\n
Let us define $\Gamma_{\alpha}^{\lambda}$, $\lambda > 0$,  as the unique self-adjoint and lower-bounded operator in $L^2(B_{a}^{c})$ associated with the quadratic form $\Phi_{\alpha}^{\lambda}$. Notice that $\Gamma_{\alpha}^\lambda$ is positive and has a bounded inverse whenever $\lambda>\lambda_0$ (with $\lambda_0$ as in Theorem \ref{thm:Qa}). We also recall that $\tau$ is the Sobolev trace operator on the coincidence hyperplane $\pi_{12}$, see~\eqref{eq:deftau}.

\begin{theorem}[Characterization of the Hamiltonian]\label{thm:thop}\newline
The self-adjoint and bounded from below operator $H_{D,\alpha}$ in $L^2_{s}(\Omega_a)$ uniquely associated to the quadratic form $Q_{D,\alpha}$ is characterized  as follows:
        \begin{gather}
            \mbox{\emph{dom}}\big(H_{D,\alpha}\big) := \left\{\psi = \varphi^{\lambda}\! + G^{\lambda} \xi \in \mbox{\emph{dom}}\big(Q_{D, \alpha}\big) \,\left|\, \varphi^{\lambda} \!\in\!  \mbox{\emph{dom}}\big(H_D\big)
                \,,\;\, \xi \in \mbox{\emph{dom}}\big(\Gamma_{\alpha}^\lambda) \,, \,  \tau \varphi^{\lambda} = \Gamma_{\alpha}^{\lambda}\xi\,, \;\, \lambda \!>\! 0 \right.\right\} , \nonumber \\
            (H_{D, \alpha} + \lambda) \psi = (H_{D} + \lambda) \varphi^{\lambda}\,.
        \end{gather}
For any $\lambda>\lambda_0$ (with $\lambda_0$ as in Theorem \ref{thm:Qa}), the associated resolvent operator $R_{D,\alpha}^{\lambda} := (H_{D, \alpha} +\lambda)^{-1}$ is given by the Krein formula
\begin{equation}\label{eq:Ral}
    R_{D,\alpha}^{\lambda} = R_D^{\lambda} + G^\lambda\, \big(\Gamma_{\alpha}^{\lambda}\big)^{-1}\, \tau R_D^{\lambda}\,.
\end{equation}
\end{theorem}

\begin{remark}\label{rem:eig}
As a consequence of the previous Theorem, one immediately sees that $\Psi_{\mu}\in \mbox{dom}\big( H_{D,\alpha}\big)$ is an eigenvector of $H_{D, \alpha}$ associated to the negative  eigenvalue $-\mu$, with $\mu>0$,  if and only if 
    \begin{equation}
        \Psi_{\mu} = G^{\mu} \xi_{\mu} \,, \qquad  
        \xi_{\mu} \in \mbox{dom} \big(\Gamma_{\alpha}^{\mu} \big)\,, \qquad 
        \Gamma_{\alpha}^{\mu} \xi_{\mu} =0\,.
    \end{equation}
\end{remark}

The last result concerns the proof of the Efimov effect in the case of infinite two-body scattering length, \emph{i.e.}, when $\alpha=0$, also known as the unitary limit.

\begin{theorem}[Efimov effect]\label{thm:theig}\newline
The Hamiltonian $H_{D,\mspace{1.5mu}0}$ has an infinite sequence of negative eigenvalues $E_n$ accumulating at zero and fulfilling
    \begin{equation}
        E_{n} = -\,\frac{4}{a^2}\, e^{\frac{2}{s_0}(\theta\mspace{1.5mu} - \mspace{1.5mu}n \pi)} \Big( 1+ o(1) \Big)\,, \qquad \mbox{for $\;n \to +\infty$}
    \end{equation}
where $\theta= \arg \Gamma(1 + i s_0)$ and $s_0\approx 1.00624$ is the unique positive solution of the equation
\begin{align}\label{eqs0}
-s \cosh \left(\tfrac{\pi}{2}\, s\right) + \tfrac{8}{\sqrt{3}}\, \sinh \left(\tfrac{\pi}{6}\, s\right) = 0\,.
\end{align}
In particular, the geometrical law~\eqref{gela} holds. Furthermore, the eigenvector associated to $E_n$ is given by
\begin{align}\label{Psin}
\Psi_{n} \big(\vec{x}, \vec{y} \big)= \psi_{n} \big(\abs{\vec{x}}, |\vec{y}| \big) 
+ \psi_{n} \!\left( \left|- \tfrac{1}{2}\, \vec{x} + \tfrac{\sqrt{3}}{2}\, \vec{y} \right|, \left|\tfrac{\sqrt{3}}{2}\, \vec{x} + \tfrac{1}{2}\, \vec{y}  \right| \right) 
+ \psi_{n} \!\left( \left|\tfrac{1}{2}\, \vec{x} + \tfrac{\sqrt{3}}{2}\, \vec{y} \right|, \left| \tfrac{\sqrt{3}}{2}\, \vec{x} - \tfrac{1}{2}\, \vec{y} \right| \right) ,
\end{align}
where 
\begin{equation}\label{psi12n}
\psi_n(r,\rho)= \frac{C_n}{4\pi \, r\mspace{1.5mu} \rho}\, \frac{\sinh \!\left(s_0 \arctan \frac{\rho}{r}\mspace{1.5mu}\right)}{\sinh \!\left(\frac{\pi}{2}\,s_0\right)} \,  K_{is_0} \!\left( \tfrac{t_n}{a} \sqrt{r^2 + \rho^2} \mspace{1.5mu}\right)\!,
\end{equation}
$r=\abs{\vec{x}}$, $\rho=|\vec{y}|$, $C_n$ is a normalization constant, $K_{is_0}$ is the modified Bessel function of the second kind with imaginary order and $t_n$ is the $n$-th positive simple root of the equation $K_{i s_0}(t)=0$. 
\end{theorem}

\section{Analysis of the quadratic form}

As a first step we derive upper and lower bounds for the quadratic form $\Phi_{\alpha}^{\lambda}[\xi]$ defined in~\eqref{phila}, \eqref{dophila}. The estimates reported in the forthcoming Lemma \ref{lemma:estPhi} ultimately account for the main results stated in Theorem \ref{thm:Qa}.

    
\begin{lemma}\label{lemma:estPhi}
    For any $\lambda > 0$ there holds
        \begin{equation}
            \Phi_{i}^{\lambda}[\xi] > 0\,, \qquad \mbox{for\; $i = 1,2,3$}\,.
        \end{equation}
    Moreover, there exist positive constants $A_{i}(\lambda)$ $(i = 1,2)$ and $B_{i}$ $(i = 1,2,3,4)$ such that:
        \begin{gather}
            A_{1}(\lambda)\, \|\xi\|^2_{L^2_{w}} \leqslant \Phi_{1}^{\lambda}[\xi] + \|\xi\|_{L^2}^{2} \leqslant B_{1}\, \|\xi\|^2_{L^2_{w}}\,; \label{eq:estPhi1} \\
            A_{2}(\lambda)\, \|\xi\|^2_{H^{1/2}} \leqslant \Phi_{2}^{\lambda}[\xi] + \|\xi\|_{L^2}^{2} \leqslant B_{2}\, \|\xi\|_{H^{1/2}}^2 \,; \label{eq:estPhi2} \\
            0 \leqslant \Phi_{3}^{\lambda}[\xi] \leqslant B_3\, \|\xi\|_{L^2}^2\,; \label{eq:estPhi3} \\
            \big|\Phi_{4}^{\lambda}[\xi]\big| \leqslant B_4\, \|\xi\|_{L^2}^2\,. \label{eq:estPhi4}
        \end{gather}
    In particular, the constants $A_{1}(\lambda)$ and $A_{2}(\lambda)$ fulfill
        \begin{equation}\label{eq:asyA1A2}
            A_1(\lambda) = o(1)\,, \qquad A_2(\lambda) = o(1)\,, \qquad\quad \mbox{for\; $\lambda \to +\infty$}\,.
        \end{equation}
\end{lemma}

\begin{proof}
We discuss separately the terms $\Phi_{i}^{\lambda}[\xi]$ ($i = 1,2,3,4$) defined in~\eqref{eq:Phi1def}~-~\eqref{eq:Phi4def}.

\textsl{1) Estimates for $\Phi_{1}^{\lambda}[\xi]$.}
Making reference to the definition~\eqref{eq:Phi1def}, we first consider the decomposition 
    \begin{align*}
        \Phi_{1}^{\lambda}[\xi] 
            = \int_{B_b \,\cap\, B_{a}^{c}}\mspace{-27mu} d\vec{y}\left(\int_{B_{a}}\!\! d\vec{y}'\, R^{\lambda}_{0}\big(\vec{0},\vec{y};\vec{0},\vec{y}'\big)\! \right) \abs{\xi(\vec{y})}^2   
            + \int_{B_{b}^{c}}\!\! d\vec{y} \left(\int_{B_{a}}\!\! d\vec{y}' \,R^{\lambda}_{0}\big(\vec{0},\vec{y};\vec{0},\vec{y}'\big)\! \right) \abs{\xi(\vec{y})}^2 .
    \end{align*}
Taking into account that the integral kernel $R^{\lambda}_{0}(\vec{X},\vec{X'})$ can be written explicitly in terms of the modified Bessel function of second kind $K_{2}$ ({\sl a.k.a.} Macdonald function), see~\eqref{eq:Gz0} in Appendix \ref{app:g0}, by direct evaluation we get
    \begin{equation*}
        R^{\lambda}_{0}\big(\vec{0},\vec{y};\vec{0},\vec{y}'\big) = {\lambda \over (2\pi)^3}\,{K_{2}\big(\sqrt{\lambda}\, \left|\vec{y} - \vec{y}'\right|\big) \over |\vec{y} - \vec{y}'|^2}\,.
    \end{equation*}
Since $t \in \mathbb{R}_{+} \mapsto K_{\nu}(t)$ is a positive definite function for any fixed $\nu \geqslant 0$ \cite[\S10.37]{NIST}, it is evident that $\Phi_{1}^{\lambda}[\xi]$ is non-negative.
Moreover, from the basic relation $\frac{\mathrm{d}}{\mathrm{d}t} \big(t^\nu K_\nu(t)\big) = -t^\nu K_{\nu-1}(t)\!<\!0$ \cite[Eq. 10.29.4]{NIST}, we deduce that $t \in \mathbb{R}_{+} \mapsto t^{\nu}\, K_{\nu}(t)$ is a continuous and decreasing function.
In particular, we have $\inf_{t \in [0,c]} t^2\, K_2(t) = c^2\, K_2(c)$ for any $c > 0$ and $\sup_{t \in \mathbb{R}_{+}\!} t^2\, K_2(t) = \lim_{t \to 0^{+}\!} t^2\, K_2(t) = 2$, see \cite[Eq. 10.30.2]{NIST}.

On one side, these arguments suffice to infer that
    \begin{equation*}
        {\lambda (a+b)^2\, K_{2}\big(\sqrt{\lambda}\, (a+b)\big) \over (2\pi)^3\,|\vec{y} - \vec{y}'|^4} \leqslant R^{\lambda}_{0}\big(\vec{0},\vec{y};\vec{0},\vec{y}'\big) \leqslant {2 \over (2\pi)^3\,|\vec{y} - \vec{y}'|^4}\,, \quad\; \mbox{for all\, $\vec{y} \!\in\! B_b \!\cap\! B_{a}^{c}$, $\vec{y}' \!\in\mspace{-1.5mu} B_{a}$}\,.
    \end{equation*}
For any $\vec{y} \in B_{a}^{c}$, a direct computation yields
    \begin{align*}
        \int_{B_{a}}\!\! d\vec{y}'\, {1 \over |\vec{y} - \vec{y}'|^4}
        = 2\pi \int_{0}^{a}\! d\rho \int_{-1}^{1}\!\!\! du\;\frac{\rho^2}{(\abs{\vec{y}}^2 + \rho^2 - 2 \abs{\vec{y}}\,\rho\,u)^2} 
        = {2\pi \over \abs{\vec{y}} \!-\! a} \left[{1 \over t \!+\! 1} + {t \!-\! 1 \over 2t} \log\!\left({t \!-\! 1 \over t \!+\! 1}\right) \right]_{t \,=\, \abs{\vec{y}}/a} \!.
    \end{align*}
It is easy to check that the expression between square brackets appearing above is a positive and decreasing function of $t \in [1,+\infty)$. As a consequence, for all $\vec{y} \!\in\mspace{-1.5mu} B_{b} \cap B_{a}^{c}$ we obtain
        \begin{align*}
        \frac{2\pi}{\abs{\vec{y}} - a} \left[\frac{1}{t + 1} + \frac{t - 1}{2t}\,\log\!\left(\frac{t - 1}{t + 1}\right) \right]_{t \,=\, b/a}
        \!\leqslant \,\int_{B_{a}}\!\! d\vec{y}'\, \frac{1}{\abs{\vec{y} - \vec{y}'}^4} 
        \,\leqslant\, \frac{\pi}{\abs{\vec{y}} - a}\,.
    \end{align*}
In view of the previous considerations, this implies
        \begin{align*}
        & \frac{1}{4\pi^2} \left[a + \frac{b^2 - a^2}{2b}\,\log\!\left(\frac{b - a}{b + a}\right) \right] \lambda\,(a+b)\,K_2\big(\sqrt{\lambda}\, (a+b)\big) \int_{B_b \,\cap\, B_{a}^{c}}\mspace{-27mu} d\vec{y}\;{\abs{\xi(\vec{y})}^2 \over \abs{\vec{y}} - a} \\
        & \hspace{4cm} \leqslant\,\int_{B_b \,\cap\, B_{a}^{c}}\mspace{-27mu} d\vec{y} \left(\int_{B_{a}}\nquad d\vec{y}'\, R^{\lambda}_{0}\big(\vec{0},\vec{y};\vec{0},\vec{y}'\big) \right) \!\abs{\xi(\vec{y})}^2 
        \leqslant\, \frac{1}{4\pi^2} \int_{B_b \,\cap\, B_{a}^{c}}\mspace{-27mu} d\vec{y}\; \frac{\abs{\xi(\vec{y})}^2\!}{\abs{\vec{y}} - a}\,.
    \end{align*}    
 \comment{\begin{align*}
        & {1 \over 4\pi^2} \left[{1 \over t + 1} + {t - 1 \over 2t}\,\log\!\left(\frac{t - 1}{t + 1}\right) \right]_{t \,=\, b/a} \nquad \lambda(a+b)^2\,K_{2}\big(\sqrt{\lambda}\, (a+b)\big) \int_{B_b \,\cap\, B_{a}^{c}}\mspace{-27mu} d\vec{y}\;{\abs{\xi(\vec{y})}^2 \over \abs{\vec{y}} - a} \\
        & \hspace{4cm} \leqslant\,\int_{B_b \,\cap\, B_{a}^{c}}\mspace{-27mu} d\vec{y} \left(\int_{B_{a}}\nquad d\vec{y}'\, R^{\lambda}_{0}\big(\vec{0},\vec{y};\vec{0},\vec{y}'\big) \right) \!\abs{\xi(\vec{y})}^2 
        \leqslant\, \frac{1}{4\pi^2} \int_{B_b \,\cap\, B_{a}^{c}}\mspace{-27mu} d\vec{y}\; \frac{\abs{\xi(\vec{y})}^2\!}{\abs{\vec{y}} - a}\,.
    \end{align*}
    }
On the other side, noting that
    \begin{equation*}
        0 \leqslant R^{\lambda}_{0}\big(\vec{0},\vec{y};\vec{0},\vec{y}'\big) \leqslant {\lambda (b-a)^2 \, K_{2}\big(\sqrt{\lambda}\, (b-a)\big) \over (2\pi)^3\,|\vec{y} - \vec{y}'|^4}\,, \qquad \mbox{for all\, $\vec{y} \!\in\! B_{b}^{c}$, $\vec{y}' \!\in\! B_{a}$}\,,
    \end{equation*}
by computation similar to those outlined above (recall, in particular, that $t^2 K_2(t) \leqslant 2$ for all $t \geqslant 0$), we infer
    \begin{align*}
        0 \leqslant \int_{B_{b}^{c}}\!\!\! d\vec{y} \left(\int_{B_{a}}\!\!\! d\vec{y}' \,R^{\lambda}_{0}\big(\vec{0},\vec{y};\vec{0},\vec{y}'\big) \right) \abs{\xi(\vec{y})}^2 
        \leqslant {1 \over 4\pi^2(b-a)} \int_{B_{b}^{c}}\hspace{-0.2cm} d\vec{y}\; \abs{\xi(\vec{y})}^2\,.
    \end{align*}
    
Summing up, we obtain   
\begin{align*}
        & \frac{1}{4\pi^2} \left[a + \frac{b^2 - a^2}{2b}\,\log\!\left(\frac{b - a}{b + a}\right) \right] \lambda\,(a+b)\,K_2\big(\sqrt{\lambda}\, (a+b)\big) \int_{B_b \,\cap\, B_{a}^{c}}\mspace{-27mu} d\vec{y}\;{\abs{\xi(\vec{y})}^2 \over \abs{\vec{y}} - a} \\
        & \hspace{5cm} \leqslant \Phi_{1}^{\lambda}[\xi] 
        \leqslant \frac{1}{4\pi^2}\! \left( \int_{B_b \,\cap\, B_{a}^{c}}\mspace{-27mu} d\vec{y}\; \frac{\abs{\xi(\vec{y})}^2\!}{\abs{\vec{y}} - a} + {1 \over b-a} \int_{B_{b}^{c}}\hspace{-0.2cm} d\vec{y}\; \abs{\xi(\vec{y})}^2 \right).
\end{align*}
\comment{\begin{align*}
        & {1 \over 4\pi^2}\! \left[{1 \over t + 1} + {t - 1 \over 2t}\,\log\!\left(\frac{t - 1}{t + 1}\right) \right]_{t \,=\, b/a} \! \lambda(a+b)^2\,K_{2}\big(\sqrt{\lambda}\, (a+b)\big) \int_{B_b \,\cap\, B_{a}^{c}}\mspace{-27mu} d\vec{y}\;{\abs{\xi(\vec{y})}^2 \over \abs{\vec{y}} - a} \\
        & \hspace{5cm} \leqslant \Phi_{1}^{\lambda}[\xi] 
        \leqslant \frac{1}{4\pi^2}\! \left( \int_{B_b \,\cap\, B_{a}^{c}}\mspace{-27mu} d\vec{y}\; \frac{\abs{\xi(\vec{y})}^2\!}{\abs{\vec{y}} - a} + {1 \over b-a} \int_{B_{b}^{c}}\hspace{-0.2cm} d\vec{y}\; \abs{\xi(\vec{y})}^2 \right).
\end{align*}
}
From here we readily deduce~\eqref{eq:estPhi1}, recalling the basic relations~\eqref{eq:wdef}~\eqref{eq:L2L2w} and noting that 
    \begin{gather*}
        \int_{B_b \,\cap\, B_{a}^{c}}\mspace{-27mu} d\vec{y}\;{\abs{\xi(\vec{y})}^2 \over \abs{\vec{y}} - a}
        = \int_{B_{a}^{c}}\!\!\! d\vec{y}\;w(\vec{y})\,\abs{\xi(\vec{y})}^2 - \int_{B_{b}^{c}}\!\!\! d\vec{y}\;\abs{\xi(\vec{y})}^2
        \geqslant \|\xi\|_{L^2_{w}}^2 - \|\xi\|_{L^2}^2\,, \\
        \int_{B_b \,\cap\, B_{a}^{c}}\mspace{-27mu} d\vec{y}\; \frac{\abs{\xi(\vec{y})}^2\!}{\abs{\vec{y}} - a} + {1 \over b-a} \int_{B_{b}^{c}}\hspace{-0.2cm} d\vec{y}\; \abs{\xi(\vec{y})}^2 
        = \tfrac{1}{b-a}\, \|\xi\|_{L^2_{w}}^{2}\,.
    \end{gather*}
The claim in~\eqref{eq:asyA1A2} regarding the constant $A_1(\lambda)$ follows by elementary considerations, noting that the map $t \mapsto t^2\, K_2(t)$ vanishes with exponential rate in the limit $t \to +\infty$ \cite[Eq. 10.40.2]{NIST}.

\textsl{2) Estimates for $\Phi_{2}^{\lambda}[\xi]$.}
Recall the explicit expression~\eqref{eq:Phi2def}. By arguments similar to those described before, it is easy to see that $\Phi_{2}^{\lambda}[\xi] \geqslant 0$. Next, let us fix arbitrarily $\varepsilon > 0$ and consider the set 
    \begin{equation*}
        \Delta_{a,\varepsilon} := \big\{(\vec{y},\vec{y}') \in B_{a}^{c} \!\times\! B_{a}^{c}\;\,\big|\;\; |\vec{y} - \vec{y}'| < \varepsilon\big\}\,.
    \end{equation*}
We re-write the definition~\eqref{eq:Phi2def} accordingly as
    \begin{align*}
        & \Phi_{2}^{\lambda}[\xi] = {1 \over 2} \int_{\Delta_{a,\varepsilon}} \nquad d\vec{y}\, d\vec{y}'\, R^{\lambda}_{0}\big(\vec{0},\vec{y};\vec{0},\vec{y}'\big)\,\big| \xi(\vec{y}) - \xi(\vec{y}') \big|^2 \\
        & \hspace{3cm} + {1 \over 2} \int_{(B_{a}^{c} \times B_{a}^{c}) \,\setminus\, \Delta_{a,\varepsilon}} \hspace{-0.6cm} d\vec{y}\, d\vec{y}'\, R^{\lambda}_{0}\big(\vec{0},\vec{y};\vec{0},\vec{y}'\big)\,\big| \xi(\vec{y}) - \xi(\vec{y}') \big|^2\,.
    \end{align*}
    
On one side, considerations analogous to those reported in part 1) of this proof yield
    \begin{equation*}
        {\lambda \varepsilon^2\, K_2\big(\sqrt{\lambda}\, \varepsilon\big) \over (2\pi)^3\,|\vec{y} - \vec{y}'|^4} \leqslant R^{\lambda}_{0}\big(\vec{0},\vec{y};\vec{0},\vec{y}'\big) \leqslant {2 \over (2\pi)^3\,|\vec{y} - \vec{y}'|^4}\,, \qquad \mbox{for all\, $(\vec{y},\vec{y}') \in \Delta_{a,\varepsilon}$}\,,
    \end{equation*}
which implies, in turn,
    \begin{align*}
        & {\lambda \varepsilon^2\, K_2\big(\sqrt{\lambda}\, \varepsilon\big) \over (2\pi)^3} \int_{\Delta_{a,\varepsilon}} \nquad d\vec{y}\, d\vec{y}'\, {\big| \xi(\vec{y}) - \xi(\vec{y}') \big|^2 \over |\vec{y} - \vec{y}'|^4} \\
        & \hspace{3cm} \leqslant \int_{\Delta_{a,\varepsilon}} \nquad d\vec{y}\, d\vec{y}'\, R^{\lambda}_{0}\big(\vec{0},\vec{y};\vec{0},\vec{y}'\big)\,\big| \xi(\vec{y}) - \xi(\vec{y}') \big|^2
        \leqslant {1 \over 4\pi^3} \int_{\Delta_{a,\varepsilon}} \nquad d\vec{y}\, d\vec{y}'\, {\big| \xi(\vec{y}) - \xi(\vec{y}') \big|^2 \over |\vec{y} - \vec{y}'|^4}\,.
    \end{align*}
    
On the other side, since
    \begin{equation*}
        0 \leqslant R^{\lambda}_{0}\big(\vec{0},\vec{y};\vec{0},\vec{y}'\big) \leqslant {2 \over (2\pi)^3\,|\vec{y} - \vec{y}'|^4}\,, \qquad \mbox{for all $(\vec{y},\vec{y}') \in (B_{a}^{c} \!\times\! B_{a}^{c}) \setminus \Delta_{a,\varepsilon}$}\,,
    \end{equation*}
we readily get
    \begin{align*}
        0 \leqslant \int_{(B_{a}^{c} \times B_{a}^{c}) \,\setminus\, \Delta_{a,\varepsilon}} \hspace{-0.3cm} d\vec{y}\, d\vec{y}'\, R^{\lambda}_{0}\big(\vec{0},\vec{y};\vec{0},\vec{y}'\big)\,\big| \xi(\vec{y}) - \xi(\vec{y}') \big|^2 
        \leqslant {1 \over 4\pi^3} \int_{(B_{a}^{c} \times B_{a}^{c}) \,\setminus\, \Delta_{a,\varepsilon}} \hspace{-0.3cm} d\vec{y}\, d\vec{y}'\, {\big| \xi(\vec{y}) - \xi(\vec{y}') \big|^2 \over |\vec{y} - \vec{y}'|^4}\,.
    \end{align*}

The above arguments imply
    \begin{align*}
    {\lambda \varepsilon^2\, K_2\big(\sqrt{\lambda}\, \varepsilon\big) \over 16\pi^3} \int_{\Delta_{a,\varepsilon}} \nquad d\vec{y}\, d\vec{y}'\, {\big| \xi(\vec{y}) - \xi(\vec{y}') \big|^2 \over |\vec{y} - \vec{y}'|^4}
    \leqslant \Phi_{2}^{\lambda}[\xi] 
    \leqslant {1 \over 8\pi^3} \int_{B_{a}^{c} \times B_{a}^{c}} \hspace{-0.3cm} d\vec{y}\, d\vec{y}'\, {\big| \xi(\vec{y}) - \xi(\vec{y}') \big|^2 \over |\vec{y} - \vec{y}'|^4}\,.
    \end{align*}
This in turn accounts for~\eqref{eq:estPhi2}, recalling the definition of the regional Gagliardo–Slobodecki\u\i~semi-norm for the Sobolev space $H^{1/2}\big(B_{a}^{c}\big)$, see~\eqref{eq:H12def}, and noting that
    \begin{align*}
        & \int_{\Delta_{a,\varepsilon}} \nquad d\vec{y}\, d\vec{y}'\, {\big| \xi(\vec{y}) - \xi(\vec{y}') \big|^2 \over |\vec{y} - \vec{y}'|^4} 
        = \int_{B_{a}^{c} \times B_{a}^{c}} \nquad d\vec{y}\, d\vec{y}'\, {\big| \xi(\vec{y}) - \xi(\vec{y}') \big|^2 \over |\vec{y} - \vec{y}'|^4} 
            - \int_{(B_{a}^{c} \times B_{a}^{c}) \,\setminus\, \Delta_{a,\varepsilon}} \nquad d\vec{y}\, d\vec{y}'\, {\big| \xi(\vec{y}) - \xi(\vec{y}') \big|^2 \over |\vec{y} - \vec{y}'|^4} \\
        & \geqslant \|\xi\|_{H^{1/2}}^{2} - \|\xi\|_{L^2}^{2}
            - 4\int_{\{\, |\vec{y} - \vec{y}'| \,>\, \varepsilon \,\}} \nquad d\vec{y}\, d\vec{y}'\, {\big| \xi(\vec{y})\big|^2 \over |\vec{y} - \vec{y}'|^4} 
        = \|\xi\|_{H^{1/2}}^{2} - \|\xi\|_{L^2}^{2}
            - 16\pi\int_{B_{a}^{c}} \!\!\! d\vec{y}\;\big| \xi(\vec{y})\big|^2 \int_{\varepsilon}^{\infty}\!\!dr\; {1 \over r^2} \\
        & = \|\xi\|_{H^{1/2}}^{2} - \left(1 + {16 \pi \over \varepsilon} \right) \|\xi\|_{L^2}^{2}\,.
    \end{align*}
Also in this case, the statement about $A_2(\lambda)$ in~\eqref{eq:asyA1A2} follows from the exponential vanishing of the function $t \mapsto t^2\, K_2(t)$ in the limit $t \to +\infty$.

\textsl{3) Estimates for $\Phi_{3}^{\lambda}[\xi]$.} Let $(r,\vec{\omega}) \in \mathbb{R}_{+} \times \mathbb{S}^2$ be a set of polar coordinates in $\mathbb{R}^3$ and let $Y_{\ell,m} \equiv Y_{\ell,m}(\vec{\omega}) \in L^2(\mathbb{S}^2)$ ($\ell \in \{0,1,2,\dots\}$, $m \in \mathbb{Z}$ with $|m| \leqslant \ell$) be the associated family of normalized spherical harmonics (see, \emph{e.g.}, \cite{avery}). In the sequel, for any $\xi \in L^2(B_{a}^{c})$ we consider the representation
    \begin{align*}
        \xi(r,\vec{\omega}) \equiv \xi\big(\vec{y}(r,\vec{\omega})\big) = \sum_{\ell \,=\, 0}^{\infty}\; \sum_{|m|\,\leqslant\, \ell}\, Y_{\ell,m}(\vec{\omega})\,\xi_{\ell,m}(r)\,,
    \end{align*}
where the coefficients $\xi_{\ell,m}$ are determined by
    \begin{align*}
        \xi_{\ell,m}(r) := \int_{\mathbb{S}^2}\!\! d\vec{\omega}\;\conjugate*{Y_{\ell,m}(\vec{\omega})}\;\xi(r,\vec{\omega})\,.
    \end{align*}
It is worth noting that
    \begin{align*}
        \|\xi\|_{L^2}^2 = \int_{B_{a}^{c}}\!\! d\vec{y}\;\abs{\xi(\vec{y})}^2 
        = \sum_{\ell \,=\, 0}^{\infty}\, \sum_{|m| \,\leqslant\, \ell}\, \int_{a}^{\infty}\!\! dr\;r^2\,\big|\xi_{\ell,m}(r)\big|^2\,.
    \end{align*}
To proceed, we refer to the explicit expression~\eqref{eq:Phi3def} for $\Phi_{3}^{\lambda}[\xi]$ and consider the series expansion for $g^{\lambda}\big(\vec{X},\vec{X}'\big)$ reported in~\eqref{eq:gzser} of Appendix \ref{app:g0}. Let us mention that the Gegenbauer polynomials $C^{2}_{\ell}$ appearing therein enjoy the following identity, for any pair of angular coordinates $\vec{\omega},\vec{\omega}' \in \mathbb{S}^2$ and any $\ell \in \{0,1,2,\ldots\}$ \cite[Eq. (66)]{avery}:
    \begin{equation}\label{eq:sumGegembauer}
        C_{\ell}^{2}(\vec{\omega} \cdot \vec{\omega}') = {4\pi \over 2 \ell + 1} \, \sum_{|m|\,\leqslant\, \ell} \conjugate*{Y_{\ell,m}(\vec{\omega}')}\; Y_{\ell,m}(\vec{\omega})\,.
    \end{equation}
On account of the above considerations and of the orthogonality of the spherical harmonics $Y_{\ell,m}$ \cite[Eq. (56)]{avery},
by direct calculations we deduce
    \begin{align*}
        \Phi_{3}^{\lambda}[\xi] 
            & = {2 \over \pi^2} \int_{a}^{\infty}\!\! dr  \int_{a}^{\infty}\!\!dr' \, 
         \sum_{\ell \,=\, 0}^{\infty} {\ell + 2 \over 2 \ell + 1} \; {I_{\ell + 2}\big(\sqrt{\lambda}\,a\big) \over K_{\ell + 2}\big(\sqrt{\lambda}\,a\big)}\,K_{\ell + 2}\big(\sqrt{\lambda}\,r\big)\,K_{\ell + 2}\big(\sqrt{\lambda}\, r'\big) \, \sum_{| m |\,\leqslant\, \ell}\,\xi^{*}_{\ell,m}(r)\,  \xi_{\ell,m}(r') \\
            & = {2 \over \pi^2} \sum_{\ell \,=\, 0}^{\infty} \, \sum_{|m| \,\leqslant\, \ell}\, {\ell + 2 \over 2 \ell + 1} \; {I_{\ell + 2}\big(\sqrt{\lambda}\,a\big) \over K_{\ell + 2}\big(\sqrt{\lambda}\, a\big)} \left| \,\int_{a}^{\infty}\!\!dr\;K_{\ell + 2}\big(\sqrt{\lambda}\,r\big)\,\xi_{\ell,m}(r) \,\right|^2 .
    \end{align*}
From here and from the positivity of the Bessel functions $K_{\nu}$ \cite[\S 10.37]{NIST}, it readily follows that $\Phi_{3}^{\lambda}[\xi] \geqslant 0$. On the other hand, by Cauchy-Schwarz inequality and the already mentioned monotonicity of the map $t \in \mathbb{R}_{+} \mapsto t^{\nu} K_{\nu}(t)$ for any fixed $\nu > 0$, we get
    \begin{align*}
        \left| \int_{a}^{\infty}\!\!dr\;K_{\ell + 2}\big(\sqrt{\lambda}\,r\big)\,\xi_{\ell,m}(r) \right|^2 \!
        & \leqslant \left({1 \over \lambda^{\ell + 2}}\, \sup_{r > a} \Big( \big(\sqrt{\lambda}\, r\big)^{\ell+2}\, K_{\ell + 2}\big(\sqrt{\lambda}\,r\big) \Big)^2\int_{a}^{\infty}\!\! dr\;{1 \over r^{2\ell+6}} \right)\! \left(\int_{a}^{\infty}\!\!\!dr\;r^2\, |\xi_{\ell,m}(r)|^2\right) \\
        & \leqslant {K^2_{\ell + 2}\big(\sqrt{\lambda}\,a\big) \over a\,(2\ell+5)} \int_{a}^{\infty}\!\!\!dr\;r^2\, |\xi_{\ell,m}(r)|^2.
    \end{align*}
Taking also into account that, for any fixed $\nu > 0$, the map $t \in \mathbb{R}_{+} \mapsto I_{\nu}(t)\,K_{\nu}(t)$ is continuous, positive and strictly decreasing with $\lim_{t \to 0^{+}} I_{\nu}(t)\,K_{\nu}(t) = {1 \over 2\nu}$ (see~\cite{baricz} and \cite[Eq. 10.29.2 and \S 10.37, together with \S 10.30(i)]{NIST}), we finally obtain
    \begin{align*}
        \Phi_{3}^{\lambda}[\xi] 
            & \leqslant {2 \over \pi^2 a} \sum_{\ell \,= \,0}^{\infty} \, \sum_{|m| \,\leqslant\, \ell} {\ell + 2 \over (2 \ell + 1)(2\ell+5)} \; I_{\ell + 2}\big(\sqrt{\lambda}\,a\big)\, K_{\ell + 2}\big(\sqrt{\lambda}\,a\big) \int_{a}^{\infty}\!\!dr\;r^2\, |\xi_{\ell,m}(r)|^2 \\
            & \leqslant {1 \over \pi^2 a} \sum_{\ell \,=\, 0}^{\infty} \, \sum_{|m| \,\leqslant\, \ell} {1 \over (2 \ell + 1)(2\ell+5)} \int_{a}^{\infty}\!\!dr\;r^2\, |\xi_{\ell,m}(r)|^2 \\
            & \leqslant {1 \over 5\pi^2 a} \sum_{\ell \,=\, 0}^{\infty} \, \sum_{|m |\,\leqslant\, \ell} \int_{a}^{\infty}\!\! dr\;r^2\, |\xi_{\ell,m}(r)|^2
            = {1 \over 5\pi^2 a}\;\|\xi\|_{L^2}^2\,,
    \end{align*}
which proves the upper bound in Eq.~\eqref{eq:estPhi3}.

\textsl{4) Estimates for $\Phi_{4}^{\lambda}[\xi]$.} Let us refer to the definition~\eqref{eq:Phi4def} and recall that $R^{\lambda}_{D}(\vec{X},\vec{X}')$ is the integral kernel associated to the resolvent operator of the Dirichlet Laplacian $H_D$ on $\Omega_a \subset \mathbb{R}^6$. The corresponding heat kernel $K_D(t;\vec{X},\vec{X}')$ is known to fulfill the following Gaussian upper bound, for all $t> 0$, $\vec{X},\vec{X}' \in \Omega_a$ and some suitable $c_1,c_2 > 0$ (see, \emph{e.g.}, \cite[p. 89, Corollary 3.2.8]{davies} and~\cite{GriSC,zhang}):
    \begin{equation*}
        0 \leqslant K_D(t;\vec{X},\vec{X}') \leqslant {c_1 \over t^3}\,e^{-{|\vec{X} - \vec{X}'|^2 \over 4 c_2 t}} .
    \end{equation*}
Taking this into account and using a well-known integral representation for the Bessel function $K_{2}$ \cite[Eq. 10.32.10]{NIST}, by classical arguments \cite[p. 101, Lemma 3.4.3]{davies} we deduce
    \begin{equation}\label{RDest}
        \big|R^{\lambda}_{D}(\vec{X},\vec{X}')\big| \leqslant \int_0^{\infty}\!\! dt\;e^{-\lambda t}\,\big|K_D(t;\vec{X},\vec{X}') \big| \leqslant c_1\! \int_0^{\infty} \! {dt \over t^3}\,e^{-\lambda t -{|\vec{X} - \vec{X}'|^2 \over 4 c_2 t}}\!
        \leqslant {8\, c_1 c_2\, \lambda  \over |\vec{X} - \vec{X}'|^2}\, K_{2}\big(\sqrt{\lambda/c_2}\, |\vec{X} - \vec{X}'|\big)\,.
    \end{equation}
Then, using the elementary inequality $\vec{y} \cdot \vec{y}' \geqslant - \big(\abs{\vec{y}}^2 + \abs{\vec{y}'}^2\big)/2$ and recalling that the map $t \in \mathbb{R}_{+} \mapsto t^2 \,K_2(t)$ is decreasing with $\lim_{t \to 0^{+}}\!t^2\, K_2(t) = 2$, we infer
    \begin{align*}
        & \left| R^{\lambda}_{D}\!\left(\vec{0},\vec{y}; {\sqrt{3} \over 2}\,\vec{y}'\!, -\,{1 \over 2}\,\vec{y}'\right) \right|
        \leqslant {8\, c_1\, c_2\, \lambda (\abs{\vec{y}}^2 + \vec{y} \cdot \vec{y}' + \abs{\vec{y}'}^2)  \over (\abs{\vec{y}}^2 + \vec{y} \cdot \vec{y}' + \abs{\vec{y}'}^2)^2}\, K_{2}\!\left(\sqrt{\lambda/c_2 (\abs{\vec{y}}^2 + \vec{y} \cdot \vec{y}' + \abs{\vec{y}'}^2})\right) \\
        & \hspace{8.5cm} \leqslant {32\, c_1\, c_2\, \lambda a^2 \, K_{2}\big(\sqrt{\lambda/c_2}\, a\big)  \over \big(\abs{\vec{y}}^2 + \abs{\vec{y}'}^2\big)^2} 
        \leqslant {64\, c_1\, c_2^2  \over \big(\abs{\vec{y}}^2 + \abs{\vec{y}'}^2\big)^2}\,.
    \end{align*}
On account of the above arguments, by Cauchy-Schwarz inequality and basic symmetry considerations, from \eqref{eq:Phi4def} we infer
    \begin{align*}
        & \big|\Phi_{4}^{\lambda}[\xi]\big| 
        = \left|\, 2 \int_{B_{a}^{c} \times B_{a}^{c}} \hspace{-0.7cm} d\vec{y}\, d\vec{y}'\, R^{\lambda}_{D}\Big(\vec{0},\vec{y}; {\sqrt{3} \over 2}\,\vec{y}'\!, -\,{1 \over 2}\,\vec{y}'\Big)\, \xi(\vec{y})\,\xi(\vec{y}') \right| \\
        & \leqslant 64\, c_1\,c_2^{2} \int_{B_{a}^{c} \times B_{a}^{c}} \hspace{-0.7cm} d\vec{y}\, d\vec{y}'\, {1 \over \big(\abs{\vec{y}}^2 + \abs{\vec{y}'}^2\big)^2}\; \abs{\xi(\vec{y})}^2 
        \leqslant 64\, c_1\,c_2^{2} \left(\,\sup_{r > a}\, \int_{B_{a}^{c}} \!\! d\vec{y}'\, {1 \over (r^2 + \abs{\vec{y}'}^2)^2} \right) \|\xi\|_{L^2}^2 \,.
    \end{align*}
This in turn implies the thesis~\eqref{eq:estPhi3}, in view of the fact that
    \begin{equation*}
        \sup_{r > a}\, \int_{B_{a}^{c}}\!\!\! d\vec{y}'\, {1 \over \big(r^2 + \abs{\vec{y}'}^2\big)^2}
        = \int_{B_{a}^{c}} \!\!\! d\vec{y}'\, {1 \over \big(a^2 + \abs{\vec{y}'}^2\big)^2}
        = 4\pi \int_{a}^{\infty}\!\!\! d\rho\;{\rho^2 \over (a^2 + \rho^2)^2} = {\pi(\pi+2) \over 2a} < \infty\,.
    \end{equation*}
\end{proof}

Building on the estimates derived in Lemma \ref{lemma:estPhi}, we can now proceed to prove Theorem \ref{thm:Qa}.

\begin{proof}[Proof of Theorem \ref{thm:Qa}]
Let us first remark that, recalling the definition~\eqref{phila} of $\Phi_{\alpha}^{\lambda}[\xi]$, the claims~\eqref{eq:upPhi} and~\eqref{eq:upPhi} are direct consequences of the upper and lower bounds reported in Lemma~\ref{lemma:estPhi}.

To proceed, we refer to the definition~\eqref{eq:Qa} of $Q_{D,\alpha}[\psi]$. From~\eqref{eq:Gz}, \eqref{eq:xisym} and~\eqref{eq:GzHs} we readily infer $\|G^{\lambda} \xi\|_{L^2} \leqslant c\,\|\xi\|_{L^2}$, where $c > 0$ is a suitable constant. We further notice that $Q_{D}[\varphi_{\lambda}] + \|\varphi_{\lambda}\|_{L^2}^{2}$ is a norm on $H^1(\Omega_a)$ equivalent to the standard one. Then, using the upper bound~\eqref{eq:upPhi} for $\big|\Phi_{\alpha}^{\lambda}[\xi]\big|$, by elementary estimates we deduce that for any $\lambda > 0$ there exists a positive constant $C_1 > 0$ such that
\begin{align*}
\big| Q_{D,\alpha}[\psi] \big| 
& = \left| Q_{D}\big[\varphi^{\lambda}\big] - 2 \lambda\, \Re \big\langle \varphi^{\lambda} \big| G^{\lambda} \xi \big\rangle - \lambda \,\|G^{\lambda}\xi \|_{L^2}^2 + 3\,\Phi_{\alpha}^{\lambda}[\xi] \right| \\
& \leqslant Q_{D}\big[\varphi^{\lambda}\big] + 2 \lambda\, \big\| \varphi^{\lambda} \big\|_{L^2}\, \big\| G^{\lambda} \xi \big\|_{L^2} + \lambda \,\|G^{\lambda}\xi \|_{L^2}^2 + 3\,\big|\Phi_{\alpha}^{\lambda}[\xi]\big| \\
& \leqslant Q_{D}\big[\varphi^{\lambda}\big] + \lambda\, \big\| \varphi^{\lambda} \big\|_{L^2}^2 + 2c\,\lambda \,\|\xi \|_{L^2}^2 + 3\,B \left(\sqrt{\lambda}\, \|\xi\|_{L^2}^2 + \|\xi\|^2_{L^2_{w}} + \|\xi\|_{H^{1/2}}^2 \right) \\
& \leqslant C_1 \Big( \big\|\varphi^{\lambda}\big\|_{H^1}^2 + \|\xi\|_{H^{1/2}}^2 + \|\xi\|^2_{L^2_{w}} \Big)\,, 
\end{align*}
This suffices to prove the well-posedness of $Q_{D,\alpha}[\psi]$ on the domain~\eqref{eq:DomQa} for any $\lambda > 0$.

Let us now show that the form does not depend on $\lambda$. To this purpose, we fix $\lambda_1 > \lambda_2 > 0$ and consider the alternative representations $\psi = \varphi^{\lambda_1} + G^{\lambda_1} \xi$ and $\psi = \varphi^{\lambda_2} + G^{\lambda_2} \xi$. From the first resolvent identity and from Eq.~\eqref{eq:GzHs}, we deduce $G^{\lambda_1} \xi - G^{\lambda_2} \xi = (\lambda_1 - \lambda_2)\, R_D^{\lambda_1} G^{\lambda_2} \xi$, which entails $G^{\lambda_1} \xi - G^{\lambda_2} \xi \in \mbox{dom}\big(H_D\big)$ for any $\xi \in H^{1/2}(B_{a}^{c}) \cap L^2_{w}(B_{a}^{c})$ (see \eqref{dophila}). In particular, we can write $\varphi^{\lambda_2} = \varphi^{\lambda_1} + G^{\lambda_1} \xi - G^{\lambda_2} \xi \in H^1_0(\Omega_a)$. Taking this into account and using the definition~\eqref{eq:Qa} of $Q_{D,\alpha}[\psi]$, by a few integration by parts we get
    \begin{equation*}
        Q_{D,\alpha}\big[\varphi^{\lambda_1} + G^{\lambda_1} \xi \big] 
        = Q_{D,\alpha}\big[\varphi^{\lambda_2} + G^{\lambda_2} \xi \big]\,,
    \end{equation*} 
which proves that the form $Q_{D,\alpha}$ is independent of $\lambda > 0$.

Finally, from the lower bound~\eqref{eq:loPhi} for $\Phi_{\alpha}^{\lambda}[\xi]$ and from the asymptotic relations reported in~\eqref{eq:asyA1A2}, for any $\lambda > \lambda_0$ we infer
    \begin{equation*}
        Q_{D,\alpha}[\psi] \geqslant Q_{D}\big[\varphi^{\lambda}\big] + \lambda \,\|\varphi^{\lambda}\|_{L^2}^2 + 3\, A_{0} \sqrt{\lambda}\, \|\xi\|_{L^2}^2 + 3\,A(\lambda) \left( \|\xi\|^2_{L^2_{w}} + \|\xi\|^2_{H^{1/2}} \right) - \lambda \,\|\psi\|_{L^2}^2 \,.
    \end{equation*}
This shows that, for any fixed $\lambda > 0$ large enough there exist two positive constants $\gamma_2,C_2 > 0$ such that
    \begin{equation}\label{eq:Qalower}
        Q_{D,\alpha}[\psi] + \gamma_{2} \,\|\psi\|_{L^2}^2 \geqslant C_2 \Big( Q_{D}\big[\varphi^{\lambda}\big] + \lambda\,\big\|\varphi^{\lambda}\big\|_{L^2}^2 + \sqrt{\lambda}\,\|\xi\|_{L^2}^{2} + \|\xi\|^2_{H^{1/2}} + \|\xi\|^2_{L^2_{w}} \Big)\,,
    \end{equation}
which proves that the form $Q_{D,\alpha}$ is coercive. From here, closedness follows as well by standard arguments~\cite{teta}. To say more, since the right-hand side of~\eqref{eq:Qalower} is clearly positive, we readily get that $Q_{D,\alpha}$ is lower-bounded.
\end{proof}

\section{The Hamiltonian}

\begin{proof}[Proof of Theorem \ref{thm:thop}]
Let us first introduce the sesquilinear form defined by the polarization identity, starting from Eq.~\eqref{eq:Qa}.
With respect to the decompositions $\psi_1 = \varphi_1^{\lambda} + G^{\lambda} \xi_1$ and $\psi_2 = \varphi_2^{\lambda} + G^{\lambda} \xi_2\,$, this is given by
\begin{align*}
Q_{D,\alpha}\big[\psi_1,\psi_2\big] & = {1 \over 4} \Big( Q_{D,\alpha}\big[\psi_1+\psi_2\big] - Q_{D,\alpha}\big[\psi_1-\psi_2\big] - i\, Q_{D,\alpha}\big[\psi_1 + i \psi_2\big] + i\, Q_{D,\alpha}\big[\psi_1 - i \psi_2\big] \Big) \\
& = Q_{D}\big[\varphi_1^{\lambda},\varphi_2^{\lambda}\big] + \lambda^2 \,\big\langle \varphi_1^{\lambda} \big| \varphi_2^{\lambda} \big\rangle - \lambda^2 \,\big\langle \psi_1 \big| \psi_2 \big\rangle + 3 \,\Phi_{\alpha}^{\lambda}\big[\xi_1,\xi_2\big]\,,
\end{align*}
where we have set
    \begin{gather*}
        Q_{D}\big[\varphi_1^{\lambda},\varphi_2^{\lambda}\big] := \int_{\Omega_a}\!\!\! d\vec{x}\,d\vec{y}\, \Big( \conjugate*{\nabla_{\vec{x}} \varphi_1^{\lambda}(\vec{x},\vec{y})} \cdot \nabla_{\vec{x}} \varphi_2^{\lambda}(\vec{x},\vec{y}) +  \conjugate*{\nabla_{\vec{y}} \varphi_1^{\lambda}(\vec{x},\vec{y})} \cdot \nabla_{\vec{y}} \varphi_2^{\lambda}(\vec{x},\vec{y}) \Big)\,, \\
        \Phi_{\alpha}^{\lambda}\big[\xi_1,\xi_2\big] := \left(\alpha + {\sqrt{\lambda} \over 4\pi}\right) \langle \xi_1 | \xi_2 \rangle + \Phi_{1}^{\lambda}\big[\xi_1,\xi_2\big] + \Phi_{2}^{\lambda}\big[\xi_1,\xi_2\big] + \Phi_{3}^{\lambda}\big[\xi_1,\xi_2\big] + \Phi_{4}^{\lambda}\big[\xi_1,\xi_2\big] \,,
    \end{gather*}   
and
    \begin{gather*} 
        \Phi_{1}^{\lambda}\big[\xi_1,\xi_2\big] := \int_{B_{a}^{c}}\!\! d\vec{y} \left(\int_{B_{a}}\!\!\! d\vec{y}'\, R^{\lambda}_{0}\big(\vec{0},\vec{y};\vec{0},\vec{y}'\big) \right) \conjugate*{\xi_1(\vec{y})\mspace{-1.5mu}}\, \xi_2(\vec{y})\,, \\
        \Phi_{2}^{\lambda}\big[\xi_1,\xi_2\big] := {1 \over 2}\int_{B_{a}^{c} \times B_{a}^{c}} \hspace{-0.4cm} d\vec{y}\, d\vec{y}'\, R^{\lambda}_{0}\big(\vec{0},\vec{y};\vec{0},\vec{y}'\big)\, \conjugate*{\big(\xi_1(\vec{y}) - \xi_1(\vec{y}')\big)}\, \big(\xi_2(\vec{y}) - \xi_2(\vec{y}')\big) \,,\\
        \Phi_{3}^{\lambda}\big[\xi_1,\xi_2\big] := - \int_{B_{a}^{c} \times B_{a}^{c}} \hspace{-0.4cm}d\vec{y}\, d\vec{y}'\, g^{\lambda}\big(\vec{0},\vec{y};\vec{0},\vec{y}'\big)\, \conjugate*{\xi_1(\vec{y})\mspace{-1.5mu}}\,\xi_2(\vec{y}')\,, \\
        \Phi_{4}^{\lambda}\big[\xi_1,\xi_2\big] := -\,2 \int_{B_{a}^{c} \times B_{a}^{c}} \hspace{-0.4cm} d\vec{y}\, d\vec{y}'\, R^{\lambda}_{D}\!\left(\vec{0},\vec{y}; \tfrac{\sqrt{3}}{2}\,\vec{y}'\!, -\,\tfrac{1}{2}\,\vec{y}'\mspace{-1.5mu}\right) \conjugate*{\xi_1(\vec{y})\mspace{-1.5mu}}\,\xi_2(\vec{y}')\,.
    \end{gather*}

Since we have already proved that the form $Q_{D,\alpha}$ is closed and lower bounded, there exists a unique associated self-adjoint and lower bounded operator $H_{D,\alpha}$. Moreover, if $\psi_2 \in \mbox{dom}\big(H_{D,\alpha}\big)$ then there exists an element $w := H_{D,\alpha} \psi_2 \in L^2(\Omega_a)$ such that $Q[\psi_1,\psi_2] = \langle \psi_1| w\rangle$ for any $\psi_1 \in \mbox{dom}\big(Q_{D,\alpha}\big)$.
    
For $\xi_1 = 0$, \emph{i.e.} $\psi_1 = \varphi_1^{\lambda} \in \mbox{dom}\big(Q_{D}\big)$, we have
    \begin{align*}
        Q_{D,\alpha}\big[\varphi_1^{\lambda},\psi_2\big] 
        = Q_{D}\big[\varphi_1^{\lambda},\varphi_2^{\lambda}\big] + \lambda \,\big\langle \varphi_1^{\lambda} \big| \varphi_2^{\lambda} \big\rangle - \lambda \,\big\langle \varphi_1^{\lambda} \big| \psi_2 \big\rangle
        = \big\langle \varphi_1^{\lambda} \big| w \big\rangle\,.
    \end{align*}
Thus, $\varphi_2^{\lambda} \in \mbox{dom}\big(H_{D}\big)$ and $w = H_{D} \varphi_2^{\lambda} - \lambda\, G^{\lambda} \xi_2$, which entails the identity
    \begin{align*}
        \big(H_{D,\alpha} + \lambda\big) \psi_2 = \big(H_{D} + \lambda\big) \varphi_2^{\lambda}\,.
    \end{align*}
    
For $\xi_1 \neq 0$, demanding that $Q_{D,\alpha}[\psi_1,\psi_2] = \langle \psi_1 | w \rangle$ with $w = H_{D} \varphi_2^{\lambda} - \lambda\, G^{\lambda} \xi_2$ as before, we obtain
    \begin{equation*}
        \big\langle G^{\lambda} \xi_1 \big| (H_{D} + \lambda)\varphi_2^{\lambda} \big\rangle
        = 3\, \Phi_{\alpha}^{\lambda}\big[\xi_1,\xi_2\big]\,.
    \end{equation*}

\n 
On account of the bosonic exchange symmetries (see~\eqref{eq:GzijEq} and~\eqref{eq:xisym}), from the basic identity~\eqref{eq:GzHs} we readily deduce the following, for all $\xi_1 \in \mbox{dom}\big(\Phi_{\alpha}^{\lambda}\big)$ and $\varphi_2^{\lambda} \in \mbox{dom}\big(H_{D}\big)$:
    \begin{equation}
        \big\langle G^{\lambda} \xi_1 \big| (H_{D} + \lambda)\,\varphi \big\rangle_{L^2(\Omega_a)}
        = 3\, \langle \xi_1 \,|\, \tau\,\varphi \rangle_{L^2(B_{a}^{c})}
        = 3 \int_{B_{a}^{c}} \!\!\! d\vec{y}\;\conjugate*{\xi_{1}(\vec{y})}\, \big(\tau \varphi_2^{\lambda}\big)(\vec{y})\,,
    \end{equation}
where $\tau$ is the Sobolev trace operator introduced in~\eqref{eq:deftau}.
Let us also remark that, since $R^{\lambda}_{0}\big(\vec{0},\vec{y}';\vec{0},\vec{y}\big)\! = \!R^{\lambda}_{0}\big(\vec{0},\vec{y};\vec{0},\vec{y}'\big)$, the definition \eqref{eq:Phi2defr} can be conveniently rephrased as
\begin{align*}
    & \Phi_{2}^{\lambda}[\xi_1,\xi_2] = \int_{B_{a}^{c} \times B_{a}^{c}} \hspace{-0.5cm} d\vec{y}\, d\vec{y}'\, R^{\lambda}_{0}\big(\vec{0},\vec{y};\vec{0},\vec{y}'\big)\, \conjugate*{\xi_1(\vec{y})}\, \big(\xi_2(\vec{y}) - \xi_2(\vec{y}')\big)\,.
\end{align*}
Summing up, we obtain
    \begin{align*}
        &\int_{B_{a}^{c}} \!\!\! d\vec{y}\;\conjugate*{\xi_{1}(\vec{y})}\, (\tau \varphi_2^{\lambda})(\vec{y})
            = \left(\alpha + \frac{\sqrt{\lambda}}{4\pi}\right)\! \int_{B_{a}^{c}}\!\!\! d\vec{y}\;\conjugate*{\xi_{1}(\vec{y})}\, \xi_2(\vec{y})
            + \int_{B_{a}^{c}}\!\!\! d\vec{y} \left(\int_{B_{a}}\!\!\! d\vec{y}'\, R^{\lambda}_{0}\big(\vec{0},\vec{y};\vec{0},\vec{y}'\big)\mspace{-4.5mu} \right) \conjugate*{\xi_1(\vec{y})\mspace{-1.5mu}}\, \xi_2(\vec{y})
        \\
        & \quad + \int_{B_{a}^{c} \times B_{a}^{c}} \hspace{-0.5cm} d\vec{y}\, d\vec{y}'\, R^{\lambda}_{0}\big(\vec{0},\vec{y};\vec{0},\vec{y}'\big)\, \conjugate*{\xi_1(\vec{y})}\, \big(\xi_2(\vec{y}) - \xi_2(\vec{y}')\big) 
            - \int_{B_{a}^{c} \times B_{a}^{c}} \hspace{-0.4cm} d\vec{y}\, d\vec{y}'\, g^{\lambda}\big(\vec{0},\vec{y};\vec{0},\vec{y}'\big)\, \conjugate*{\xi_1(\vec{y})}\,\xi_2(\vec{y}')
        \\
        & \quad -2 \int_{B_{a}^{c} \times B_{a}^{c}} \hspace{-0.4cm} d\vec{y}\, d\vec{y}'\, R^{\lambda}_{D}\!\left(\vec{0},\vec{y}; \tfrac{\sqrt{3}}{2}\,\vec{y}'\!, -\,\tfrac{1}{2}\,\vec{y}'\right) \conjugate*{\xi_1(\vec{y})}\,\xi_2(\vec{y}')\,.
    \end{align*}
Then the thesis follows by the arbitrariness of $\xi_1$, recalling the definition of $\Gamma_{\alpha}^{\lambda}$ (see the comments reported after Theorem \ref{thm:Qa}).
Finally, the representation~\eqref{eq:Ral} for the resolvent operator $R_\alpha^{\lambda}$ follows by general resolvent identities, noting that $\Gamma_{\alpha}^{\lambda}$ has a bounded inverse for $\lambda > 0$ large enough (see, \emph{e.g.}, \cite{CFP,posilicano}).
\end{proof}

In view of the arguments reported in the above proof, it appears that the action of the operator $\Gamma_{\alpha}^{\lambda}$ on any $\xi \in \mbox{dom}\big(\Gamma_{\alpha}^{\lambda}\big)$ is given by
\begin{align}
        \big(\Gamma_{\alpha}^{\lambda} \xi\big)(\vec{y})
        & = \left(\alpha + \frac{\sqrt{\lambda}}{4\pi}\right)\! \xi(\vec{y})
            + \xi(\vec{y}) \int_{B_{a}}\!\!\!\! d\vec{y}'\, R^{\lambda}_{0}\big(\vec{0},\vec{y};\vec{0},\vec{y}'\big)
            + \int_{B_{a}^{c}}\!\!\!\! d\vec{y}'\, R^{\lambda}_{0}\big(\vec{0},\vec{y};\vec{0},\vec{y}'\big)\, \big(\xi(\vec{y}) - \xi(\vec{y}')\big) 
            \nonumber \\
        & \qquad 
            - \int_{B_{a}^{c}}\!\!\! d\vec{y}'\, g^{\lambda}\big(\vec{0},\vec{y};\vec{0},\vec{y}'\big)\, \xi(\vec{y}')
            - 2\int_{B_{a}^{c}}\!\!\! d\vec{y}' \, R^{\lambda}_{D}\!\left(\vec{0},\vec{y}; \tfrac{\sqrt{3}}{2}\,\vec{y}'\!, -\,\tfrac{1}{2}\,\vec{y}'\right) \xi(\vec{y}')\,. \label{Gammaexp}
    \end{align}

\begin{lemma}\label{GammaDomain}
For any $\lambda>0$, there holds $\mbox{\emph{dom}}\big(\Gamma_{\alpha}^{\lambda}\big) = H^1_0(B_{a}^{c})$.
\end{lemma}

\begin{proof}
We refer to the explicit expression~\eqref{Gammaexp} for $\Gamma_{\alpha}^{\lambda} \xi$, regarding it as a sum of five distinct terms. The first linear addendum is trivially defined on the whole space $L^2(B_{a}^{c})$. Hereafter we discuss separately the remaining terms.

\emph{1) On the second term in~\eqref{Gammaexp}.} By arguments analogous to those described in the proof of Lemma \ref{lemma:estPhi} we get
    \begin{equation*}
        \int_{B_{a}^{c}}\!\!\! d\vec{y}\, \left\lvert\,\int_{B_a}\!\!\! d\vec{y}'\, R^{\lambda}_{0}\big(\vec{0},\vec{y};\vec{0},\vec{y}'\big) \right\rvert^2 \!\abs{\xi(\vec{y})}^2\leqslant\frac{1}{16\pi^4}\int_{B_{a}^{c}}\!\!\! d\vec{y}\; \frac{\abs{\xi(\vec{y})}^2}{(\abs{\vec{y}}-a)^2} \,.
    \end{equation*}
Recalling that, for any $\xi \!\in\! H^1(B_{a}^{c})$, one has $\xi\!\in\! H^1_0(B_{a}^{c})$ if and only if $\frac{\xi}{\abs{\vec{x}}-a}\!\in\! L^2(B_{a}^{c})$ \cite[p. 74, Example 9.12]{kufner}, the above estimate proves that the second term in~\eqref{Gammaexp} is a bounded operator from $H^1_0(B_{a}^{c})$ into $L^2(B_{a}^{c})$.

\emph{2) On the third term in~\eqref{Gammaexp}.} Let us consider the decomposition
    \begin{equation}\label{eq:third}
        \int_{B_{a}^{c}}\!\!\! d\vec{y}'\,R_0^{\lambda}(\vec{0},\vec{y};\vec{0},\vec{y}')\big(\xi(\vec{y})-\xi(\vec{y}')\big)
        = \frac{1}{4\pi}\, (\mathscr{L}\xi)(\vec{y})
            - \frac{1}{4\pi^3}\int_{B_{a}^{c}}\!\!\! d\vec{y}'\,\mathscr{H}^\lambda(\vec{y}-\vec{y}')\,\big(\xi(\vec{y})-\xi(\vec{y}')\big)\,,
    \end{equation}
where we have set
    \begin{gather*}
        (\mathscr{L} \xi)(\vec{y}) := \frac{1}{\pi^2}\int_{B_{a}^{c}}\!\!\! d\vec{y}'\;\frac{\xi(\vec{y})-\xi(\vec{y}')}{|\vec{y}-\vec{y}'|^4}\,, \hspace{1.5cm}
        \mathscr{H}^{\lambda}(\vec{y}) := \frac{1}{\abs{\vec{y}}^4}\left[1-\tfrac{1}{2}\,t^2 \,K_2(t) \right]_{t \,=\,\sqrt{\lambda}\, |\vec{y}|} .
    \end{gather*}
On one side, for any given $\xi\in H^1_0(B_{a}^{c})$ we consider the extension $\tilde{\xi}\in H^1(\mathbb{R}^3)$ such that $\tilde{\xi} = \xi$ a.e. in $B_{a}^{c}$ and $\tilde{\xi}=0$ a.e. in $B_{a}$ \cite[Thm. 11.4]{LM}. Then, by an explicit calculation reported in the proof of Lemma \ref{lemma:estPhi}, for a.e. $\vec{y} \in B_{a}^{c}$ we obtain
    \begin{align*}
        (\mathscr{L} \xi)(\vec{y})
        & = \frac{1}{\pi^2}\int_{\mathbb{R}^3}\!\!\! d\vec{y}'\;\frac{\tilde{\xi}(\vec{y})-\tilde{\xi}(\vec{y}')}{|\vec{y}-\vec{y}'|^4} 
            - \left(\frac{1}{\pi^2}\int_{B_{a}}\!\!\! d\vec{y}'\;\frac{1}{|\vec{y}-\vec{y}'|^4}\right) \xi(\vec{y})\\
        & = \left[(-\Delta)^{1/2}\,\tilde{\xi}\right](\vec{y}) - \frac{2}{\pi}\!\left[\frac{1}{t+1}+\frac{t-1}{2t}\ln\!\left(\frac{t-1}{t+1}\right)\!\right]_{t \,=\, |\vec{y}|/a}\, {\xi(\vec{y}) \over |\vec{y}| - a}\,,\\
    \end{align*}    
where $(-\Delta)^{1/2} : H^1(\mathbb{R}^3) \to L^2(\mathbb{R}^3)$ is the square root of the Laplacian \cite[\S 3]{hitch}. Keeping in mind that the function between square brackets is upper and lower bounded, and recalling again that $\frac{\xi}{\abs{\vec{y}}\,-\,a}\in L^2(B_{a}^{c})$ for $\xi \in H^{1}(B_{a}^{c})$ if and only if $\xi \in H^{1}_{0}(B_{a}^{c})$, we see that $\mathscr{L}$ is a bounded operator from $H^1_0(B_{a}^{c})$ to $L^2(B_{a}^{c})$.
On the other side, noting that the function $h(t) := 1- \tfrac{1}{2}\,t^2\, K_2(t)$ ($t > 0$) is strictly increasing with $h'(t) = \tfrac{1}{2}\,t^2\, K_1(t)$, $\lim_{t \to 0^{+}} h(t)/t = 0$ and $\lim_{t \to +\infty} h(t) = 0$ \cite[Eqs. 10.29.4, 10.30.2 and 10.25.3]{NIST}, we infer by direct inspection that $\mathscr{H}^{\lambda}\in L^1(\mathbb{R}^3)$.\footnote{
More precisely, integrating by parts and using a known integral identity for the Bessel function \cite[p. 676, Eq. 6.561.16]{GR}, we obtain
$$ 
    \| \mathscr{H}^\lambda \|_{L^1}
    = \int_{\mathbb{R}^3}\!\!\! d\vec{y}\;\frac{1}{\abs{\vec{y}}^4}\left[1-\tfrac{1}{2}\,t^2 \,K_2(t) \right]_{t \,=\,\sqrt{\lambda}\, |\vec{y}|} \!
    = 4\pi \sqrt{\lambda} \int_{0}^{\infty}\!\!\! dt\;{h(t) \over t^2}
    = 4\pi \sqrt{\lambda} \int_{0}^{\infty}\!\!\! dt\;\frac{h'(t)}{t}
    = 2\pi \sqrt{\lambda} \int_{0}^{\infty}\!\!\! dt\;t\, K_1(t)
    = \pi^2 \sqrt{\lambda} \;.
$$}
As a consequence, by elementary estimates and Young's convolution inequality \cite[Thm. 4.5.1]{hormander} we infer
\begin{align*}
    & \int_{B_{a}^{c}}\!\!\!d\vec{y} \left|\int_{B_{a}^{c}}\!\!\! d\vec{y}'\,\mathscr{H}^\lambda(\vec{y}-\vec{y}')\,\big(\xi(\vec{y})-\xi(\vec{y}')\big) \right|^2 \\
    & \leqslant 2 \int_{B_{a}^{c}}\!\!\!d\vec{y} \left|\int_{B_{a}^{c}}\!\!\! d\vec{y}'\,\mathscr{H}^\lambda(\vec{y}-\vec{y}')\, \right|^2 \big|\xi(\vec{y})\big|^2 + 2 \int_{B_{a}^{c}}\!\!\!d\vec{y} \left|\int_{B_{a}^{c}}\!\!\! d\vec{y}'\,\mathscr{H}^\lambda(\vec{y}-\vec{y}')\,\xi(\vec{y}') \right|^2 \\
    & \leqslant 4\,\|\mathscr{H}^\lambda\|_{L^1}^{2}\, \|\xi\|_{L^2}^2 \,.
\end{align*}
Summing up, the previous results allow us to infer that the third term in~\eqref{Gammaexp} defines a bounded operator from $H^1_0(B_{a}^{c})$ to $L^2(B_{a}^{c})$.

\emph{3) On the fourth term in~\eqref{Gammaexp}.} Retracing the arguments described in step \emph{3)} of the proof of Lemma \ref{lemma:estPhi}, it can be shown that the term under analysis is a bounded operator in $L^2(B_{a}^{c})$. More precisely, decomposing $\xi$ into spherical harmonics $Y_{\ell,m}$ and exploiting the explicit representation~\eqref{eq:gzser} for $g^{\lambda}(\vec{X},\vec{X}')$, from the summation formula~\eqref{eq:sumGegembauer} and the orthonormality of the spherical harmonics we deduce
\begin{align*}
    & \int_{B_{a}^{c}}\!\!\! d\vec{y}\, \left\lvert\,\int_{B_{a}^{c}}\nquad d\vec{y}'\,g^{\lambda}(\vec{0},\vec{y};\vec{0},\vec{y}')\xi(\vec{y}')\right\rvert^2 \\
    & =\frac{4}{\pi^4}\int_{a}^{\infty}\!\! dr\;\!\int_{\mathbb{S}^2}\mspace{-6mu}d\vec{\omega}\;\!\left\lvert\sum_{\ell \,=\, 0}^{\infty} \frac{\ell+2}{2\ell+1}\, \frac{I_{\ell+2}\big(\sqrt{\lambda}\, a\big)}{K_{\ell+2}\big(\sqrt{\lambda}\, a\big)}\frac{K_{\ell+2}\big(\sqrt{\lambda}\, r\big)}{r}\!\sum_{|m|\,\leqslant\,\ell} Y_{\ell,m}(\vec{\omega})\int_{a}^{\infty}\!\!\! dr'\;K_{\ell+2}\big(\sqrt{\lambda}\, r'\big)\, \xi_{\ell, m}(r')\right\rvert^2\\
    & = \frac{4}{\pi^4}\int_{a}^{\infty}\!\!\! dr\, \sum_{\ell \,=\, 0}^{\infty} \frac{(\ell+2)^2}{(2\ell+1)^2}\, \frac{I_{\ell+2}^2 \big(\sqrt{\lambda}\, a\big)}{K_{\ell+2}^2 \big(\sqrt{\lambda}\, a\big)}\frac{K_{\ell+2}^2\big(\sqrt{\lambda}\, r\big)}{r^2} \sum_{|m| \, \leqslant\, \ell} \left\lvert\,\int_{a}^{\infty}\!\!\! dr'\;K_{\ell+2}\big(\sqrt{\lambda}\, r'\big)\, \xi_{\ell, m}(r')\right\rvert^2 .
\end{align*}
Then, recalling that the function $t \mapsto t^{\nu}\, K_{\nu}(t)$ ($t \in \mathbb{R}_{+}$, $\nu > 0$) is decreasing and noting that the map $t \mapsto I_{\nu}(t)\,K_{\nu}(t)$  ($t \in \mathbb{R}_{+}$, $\nu > 0$) is decreasing as well with $\lim_{t \to 0^{+}} I_{\nu}(t)\,K_{\nu}(t) = {1 \over 2\nu}$ (see~\cite{baricz} and \cite[Eq. 10.29.2 and \S 10.37, together with \S 10.30(i)]{NIST}), we obtain
\begin{align*}
    & \int_{B_{a}^{c}} \nquad d\vec{y}\, \left\lvert\,\int_{B_{a}^{c}}\nquad d\vec{y}'\,g^{\lambda}(\vec{0},\vec{y};\vec{0},\vec{y}')\xi(\vec{y}')\right\rvert^2 \\
    &\leqslant \frac{4}{\pi^4}\int_{a}^{\infty}\!\!\! dr\;\!\sum_{\ell \,=\,0}^{\infty} \frac{(\ell+2)^2}{(2\ell+1)^2}\frac{I_{\ell+2}^2\big(\sqrt{\lambda}\, a\big)}{K_{\ell+2}^2\big(\sqrt{\lambda}\, a\big)}\frac{K_{\ell+2}^2 \big(\sqrt{\lambda}\, r\big)}{r^2} \sum_{|m|\,\leqslant\, \ell} \left(\int_{a}^{\infty}\!\!\! dr'\;\frac{(r')^{2\ell+4}\,K_{\ell+2}^2 \big(\sqrt{\lambda}\, r'\big)}{(r')^{2\ell+6}}\right) \left( \int_{a}^{\infty}\!\!\! dr'\,(r')^2\,\abs{\xi_{\ell, m}(r)}^2 \right)\\
    &\leqslant \frac{4}{\pi^4 a}\int_{a}^{\infty}\!\!\! dr\;\!\sum_{\ell\, = \,0}^{\infty} \frac{(\ell+2)^2}{(2\ell+1)^2(2\ell+5)}\,I_{\ell+2}^2\big(\sqrt{\lambda}\, a\big) \frac{K_{\ell+2}^2 \big(\sqrt{\lambda}\, r\big)}{r^2} \sum_{|m| \,\leqslant \ell} \int_{a}^{\infty}\!\!\! dr'\;(r')^2\abs{\xi_{\ell, m}(r')}^2\\
    &\leqslant \frac{4}{\pi^4 a}\int_{a}^{\infty}\!\!\! dr\;\!\sum_{\ell \,=\,0}^{\infty}  \frac{(\ell+2)^2}{(2\ell+1)^2(2\ell+5)}\,\frac{I_{\ell+2}^2\big(\sqrt{\lambda}\, a\big)\, (\lambda a^2)^{\ell + 2} K_{\ell+2}^2 \big(\sqrt{\lambda}\, a\big)}{(\lambda r^2)^{\ell + 2}\, r^{2}} \sum_{|m| \,\leqslant\, \ell} \int_{a}^{\infty}\!\!\! dr'\;(r')^2\abs{\xi_{\ell, m}(r')}^2\\
    &\leqslant \frac{1}{\pi^4}\int_{a}^{\infty}\!\!\! dr\;\!\sum_{\ell \,=\,0}^{\infty}  \frac{1}{(2\ell+1)^2(2\ell+5)}\;\frac{a^{2\ell + 3}}{r^{2\ell + 6}} \sum_{|m| \,\leqslant\, \ell} \int_{a}^{\infty}\!\!\! dr'\;(r')^2\abs{\xi_{\ell, m}(r')}^2\\
    & = \frac{1}{\pi^4 a^2} \sum_{\ell \,=\,0}^{\infty} \frac{1}{(2\ell+1)^2(2\ell+5)^2} \sum_{|m| \,\leqslant\, \ell} \int_{a}^{\infty}\!\!\! dr'\;(r')^2\abs{\xi_{\ell, m}(r')}^2\\
    & \leqslant \frac{1}{25 \pi^4 a^2} \sum_{\ell \,=\,0}^{\infty} \sum_{|m| \,\leqslant\, \ell} \int_{a}^{\infty}\!\!\! dr'\;(r')^2\abs{\xi_{\ell, m}(r')}^2
        = \frac{1}{25 \pi^4 a^2} \, \|\xi\|^2_{L^2}\, .
\end{align*}

\emph{4) On the fifth term in~\eqref{Gammaexp}.} This term identifies a bounded operator in $L^2(B_{a}^{c})$. We proceed to justify this claim using arguments analogous to those reported in step \emph{4)} of the proof of Lemma \ref{lemma:estPhi}. More precisely, exploiting the upper bound~\eqref{RDest} for the off-diagonal resolvent kernel, we deduce
\begin{align*}
    & \int_{B_{a}^{c}}\!\!\! d\vec{y}\, \left\lvert\,\int_{B_{a}^{c}}\!\!\! d\vec{y}'\,R_D^{(-\lambda^2)}\!\left(\vec{0},\vec{y};\frac{\sqrt{3}}{2}\vec{y}',-\frac{1}{2}\vec{y}'\right)\!\xi(\vec{y}')\right\rvert^2 
        \leqslant B(\lambda)^2 \int_{B_{a}^{c}}\!\!\! d\vec{y}\, \left(\,\int_{B_{a}^{c}}\!\!\! d\vec{y}'\,\frac{\abs{\xi(\vec{y}')}}{(\abs{\vec{y}}^2+\abs{\vec{y}'}^2)^2}\right)^2\\
    &\leqslant B(\lambda)^2 \int_{B_{a}^{c}}\!\!\! d\vec{y}\, \left(\int_{B_{a}^{c}}\!\!\! d\vec{z}\;\frac{\abs{\xi(\vec{z})}^2}{(\abs{\vec{y}}^2+\abs{\vec{z}}^2)^2} \right) \left(\int_{B_{a}^{c}}\!\!\! d\vec{z}'\;\frac{1}{(\abs{\vec{y}}^2+\abs{\vec{z}'}^2)^2}\right) \\
    &\leqslant B(\lambda)^2 \left( \int_{B_{a}^{c}}\!\!\! d\vec{y}\;\frac{1}{(\abs{\vec{y}}^2+a^2)^2} \right) \left(\int_{B_{a}^{c}}\!\!\! d\vec{z}'\;\frac{1}{(a^2+\abs{\vec{z}'}^2)^2}\right) \|\xi\|^2_{L^2}
    = B(\lambda)^2\,\frac{\pi^2(\pi+2)^2}{4a^2}\;\|\xi\|^2_{L^2}\,.
\end{align*}
\end{proof}

\section{Efimov effect in the unitary limit}

Let us consider the eigenvalue problem in the case $\alpha=0$, corresponding to the case of infinite two-body scattering length, also known as the {\sl unitary limit}. In this section we show that the Hamiltonian $H_{D,0}$ has an infinite sequence of negative eigenvalues accumulating at zero and satisfying the Efimov geometrical law~\eqref{gela}.

\n
The first step is the construction of a sequence of eigenvectors and eigenvalues at a formal level, following the standard procedure used in the physical literature (see \emph{e.g.}~\cite{NE}). 

\n
We write a generic eigenvector associated to the negative eigenvalue $-\mu$ ($ \mu>0$) in the form
\begin{equation}\label{aut1}
\Psi_{\mu} \big(\vec{x}, \vec{y} \big)= \psi_{\mu} \big(\vec{x}, \vec{y} \big) + \psi_{\mu} \!\left(\! - \tfrac{1}{2}\, \vec{x} + \tfrac{\sqrt{3}}{2}\, \vec{y}, -\tfrac{\sqrt{3}}{2}\, \vec{x} - \tfrac{1}{2}\, \vec{y}\! \right) + \psi_{\mu} \!\left(\! - \tfrac{1}{2}\, \vec{x} - \tfrac{\sqrt{3}}{2}\, \vec{y}, \tfrac{\sqrt{3}}{2}\, \vec{x} - \tfrac{1}{2}\, \vec{y}\right),
\end{equation}

\n
where $\psi_{\mu} = G^{\mu}_{12} \xi_{\mu}$, for some suitable $\xi_{\mu} \in L^2(B_{a}^{c})$, so that $\Psi_\mu$ is decomposed in terms of the so-called Faddeev components.
To simplify the notation, from now on we drop the dependence on $\mu$. We look for $\psi$ depending  only on the radial variables $r=\abs{\vec{x}}$ and $\rho=\abs{\vec{y}}$, so with an abuse of notation we set $\psi = \psi(r, \rho)$. Hence we have
\begin{equation}\label{eqra1}
\big(- \Delta_{\vec{x}} - \Delta_{\vec{y}} + \mu \big) \psi 
= -\, \frac{1}{r^2} \frac{\partial}{\partial r} \left( r^2\, \frac{\partial \psi}{\partial r} \right) 
    -\frac{1}{\rho^2} \frac{\partial}{\partial \rho} \left( \rho^2\, \frac{\partial \psi}{\partial \rho} \right) 
    + \mu \,\psi 
= 0\,,   \qquad \mbox{in\, $D_a$}\,,
\end{equation}
where 
\begin{equation}
    D_a= \big\{ (r,\rho)\in \mathbb{R}_{+} \times \mathbb{R}_{+} \;\big|\; r^2 + \rho^2 >a^2 \big\}\,.
\end{equation}

\n
Moreover, we impose the Dirichlet boundary condition 
\begin{align}\label{dirbc}
\psi(r,\rho)=0 \qquad \mbox{in\; $\big\{(r,\rho)\in \partial D_a \,\big|\,r^2 + \rho^2 = a^2\big\}$}\,,
\end{align}
and the (singular) boundary condition $\Psi = \frac{\xi(\rho)}{4\pi\,r} + o(1)$ for $r \to 0$, see \eqref{eq:deltabc}. Taking into account that $\xi(\rho)= 4\pi\,(r \,\psi )(0, \rho)$, in view of~\eqref{aut1} the boundary condition reads
\begin{equation}\label{bc1}
\lim_{r \to 0} \left[ \psi (r,\rho) - \frac{(r\, \psi)(0,\rho)}{r} \right]\mspace{-1.5mu} + 2 \,\psi\!\left( \tfrac{\!\sqrt{3}}{2}\mspace{1.5mu}\rho, \tfrac{1}{2}\mspace{1.5mu}\rho \mspace{-1.5mu}\right)=0\,, \qquad \mbox{for\, $\rho>a$}\,.
\end{equation}

\n
It turns out that the function $\psi$ can be explicitly determined as a solution in $L^2(\mathbb{R}^6)$ to the boundary value problem~\eqref{eqra1}, \eqref{dirbc}, \eqref{bc1}.
We outline the construction in appendix~\ref{bvp} for convenience of the reader. Here we simply state the result.
Let $K_{is_0}: \mathbb{R}_{+} \to \mathbb{R}$ be the modified Bessel function of the second kind with imaginary order, where $s_0$ is the unique positive solution of~\eqref{eqs0}.

\n
Let $\{t_n\}_{n \in \N}$ be the sequence of positive simple roots of the equation $K_{is_0}(t)=0$, where $t_n \to 0$ for $n \to +\infty$. Taking into account that the asymptotic expansion of $K_{i s_0}(t)$ for $t \to 0$ is given by \cite[Eq. 10.45.7]{NIST}
\begin{equation}
    K_{i s_0}(t) = - \sqrt{\tfrac{\pi}{s_0 \sinh (\pi s_0)}} \sin \big( s_0 \log \tfrac{t}{2} - \theta \big) + O(t^2) \,, \qquad \theta= \arg \Gamma(1 + i s_0)\,,
\end{equation}
one also has the following asymptotic behavior 
\begin{align}
t_n= 2 \, e^{\frac{\theta - n \pi}{s_0}} (1 + \epsilon_n)\,, \qquad\quad \mbox{with\, $\epsilon_n \!\to\! 0$\; for\; $n \!\to\! +\infty$}\,. 
\end{align}

Then we have that, for each
\begin{align}\label{mun}
\mu= \mu_n= \left( \frac{t_n}{a} \right)^{\!2}\,,
\end{align}
the boundary value problem~\eqref{eqra1}, \eqref{dirbc}, \eqref{bc1} has a solution in $L^2(\mathbb{R}^6)$ given by~\eqref{psi12n}, \emph{i.e.},
\begin{equation*}
\psi_n(r,\rho) = \frac{C_n}{4\pi\,r \rho}\, \frac{\sinh \left( s_0 \arctan \frac{\rho}{r} \right)}{\sinh \left(s_0 \tfrac{\pi}{2}\right)} \, K_{is_0} \!\left( \tfrac{t_n}{a} \sqrt{r^2 + \rho^2} \right),
\end{equation*}
where $C_n$ is an arbitrary constant. We define the charge distribution associated to $\psi_n$ as
\begin{equation}\label{xindef}
\xi_{n}(\rho) := 4\pi \lim_{r \to 0^{+}} r\, \psi_n(r,\rho) = \frac{C_n}{\rho}\,  K_{is_0} \!\left( \tfrac{t_n}{a}\, \rho \right) \qquad (\rho > a)\,.
\end{equation}
Let us stress that $\xi_{n}$ actually keeps track of the Dirichlet boundary condition~\eqref{dirbc} for $\psi_n$; in fact, we have $\xi_{n}(a) = {C_n \over a}\,K_{i s_0}(t_n) = 0$. With a slight abuse of notation, in the sequel we refer to the function
    \begin{equation*}
        \xi_n(\mathbf{y}) \equiv \xi_n\big(|\mathbf{y}|\big) \in L^2(B_{a}^{c})\,.
    \end{equation*}

The next step is to show that the function $\psi_n$ can be written as the potential generated by the charge~\eqref{xindef} distributed on the hyperplane $\pi_{12}$, \emph{i.e.} to prove the following lemma.

\begin{lemma} 
For any fixed $n \in \{1,2,3, \dots \}$, let $\Psi_n$ and $\xi_n$ be as in Eqs.~\eqref{Psin} and~\eqref{xindef}, respectively.
Then, $\xi_n \!\in\! \mbox{dom}\big(\Gamma_{0}^{\mu_n}\big)$ and 
\begin{equation}\label{eq:FaddeevComponent}
\Psi_n = G^{\mu_n} \xi_n\,. 
\end{equation}
\end{lemma}

\begin{proof}
Many of the arguments presented in this proof rely on direct inspection of the explicit expressions~\eqref{Psin}, \eqref{psi12n} and~\eqref{xindef}. In particular, we shall often refer to a well-known integral representation of the Bessel function $K_{i s_0}$, namely \cite[Eq. 10.32.9]{NIST}
\begin{equation}\label{Kint}
    K_{i s_0}(t) = \int_{0}^{\infty}\!\! dz\;\, cos(s_0\mspace{1.5mu} z)\,e^{- \, t \cosh z} \qquad (t > 0)\,.
\end{equation}

Firstly, using~\eqref{Kint} it is easy to check that $K_{i s_0}$ is smooth on the (open) positive real semi-axis and that it vanishes with exponential rate at infinity, together with all its derivatives. This ensures, in particular, that $\xi_n\!\in\! H^1\big(B_{a}^{c}\big)$. To say more, given that all the zeros of $K_{is_0}$ are simple~\cite[\S 10.21(i)]{NIST}, we have $\frac{\xi_n(\rho)}{\rho\, - \mspace{1.5mu}a}\in L^2 (B_{a}^{c})$. In view of \cite[Example 9.12]{kufner} and of the previous considerations, we deduce that $\xi_n \in H^1_0(B_{a}^{c})$, which implies in turn $\xi_{n} \in \mbox{dom}\big(\Gamma_{0}^{\lambda}\big) \subset \mbox{dom}\big(\Phi_{0}^{\lambda}\big)$ by Lemma \ref{GammaDomain}. Incidentally, we remark that $G^{\mu_n} \xi_n$ is well defined since $\xi_n \!\in\! \mbox{dom}\big(\Phi_{0}^{\lambda}\big)$, see~\eqref{eq:Gz} and \eqref{eq:GzHs}.
    
Let us now proceed to prove~\eqref{eq:FaddeevComponent}, to be regarded as an identity of elements in $L^2_{s}(\Omega_{a})$. To this avail it suffices to show that, for all $\varphi \in \mbox{dom}\big(H_{D}\big)$, there holds
    \begin{equation*}
        \big\langle \Psi_n \big| (H_D + \mu_n) \varphi \big\rangle
        = \big\langle G^{\mu_n} \xi_n \big| (H_D + \mu_n) \varphi \big\rangle \,.
    \end{equation*}
Noting that $\mu_{n} \!<\! 0$ belongs to the resolvent set of $H_{D}$, from~\eqref{eq:GzHs} we deduce $\big\langle G^{\mu_n} \xi_n \big| (H_D + \mu_n) \varphi \big\rangle = 3\, \langle \xi_n \big| \tau \varphi \rangle$, where $\tau$ is the Sobolev trace on $\pi_{12}$ (see~\eqref{eq:deftau}). Then, the thesis follows as soon as we prove that, for all $\varphi \!\in\! \mbox{dom}\big(H_{D}\big)$,
    \begin{equation*}
        \big\langle \Psi_n \big| (H_D + \mu_{n}) \varphi \big\rangle
        = 3\, \langle \xi_n | \tau \varphi \rangle\,.
    \end{equation*}
Equivalently, due to the bosonic symmetry (see~\eqref{Psin}), we must show that 
    \begin{equation}\label{eq:thesis}
        \big\langle \psi_n \big| (H_D + \mu_{n}) \varphi \big\rangle
        = \langle \xi_n | \tau \varphi \rangle\,.
    \end{equation}

As an intermediate step, we henceforth derive~\eqref{eq:thesis} for all $\varphi \in \mathcal{D}(\Omega_a)$, where
    \begin{equation*}
        \mathcal{D}(\Omega_a) := \Big\{ \varphi \!\in\! C^{\infty}\big(\,\overline{\Omega}_a\big)\;\Big|\; \varphi\!\upharpoonright\! \partial \Omega_a = 0\,, \; \mbox{with\, $\mbox{supp}\mspace{1.5mu} \varphi$\, compact} \Big\}\, .
    \end{equation*}
Using Green's second identity, we infer
    \begin{align*}
        & \big\langle \psi_n \big| (H_D + \mu_{n}) \varphi \big\rangle
        = \lim_{\varepsilon \to 0^{+}} \int_{\Omega_a \setminus\, \pi_{12}^{\varepsilon}} \hspace{-0.6cm} d\vec{x}\,d\vec{y}\; \overline{\psi_n}\, (- \Delta_{\vec{x}} - \Delta_{\vec{y}} + \mu_{n}) \varphi \\
        & = \lim_{\varepsilon \to 0^{+}}\! \left[ \int_{\Omega_{a} \setminus\, \pi_{12}^{\varepsilon}} \hspace{-0.9cm} d\vec{x}\,d\vec{y}\; \overline{(- \Delta_{\vec{x}} - \Delta_{\vec{y}} + \mu_{n}) \psi_n}\; \varphi 
            + \int_{\partial \pi_{12}^{\varepsilon}} \hspace{-0.5cm} d\Sigma\, \Big( \overline{\partial_{\vec{\nu}} \psi_n}\, \varphi - \overline{\psi_n}\, \partial_{\vec{\nu}} \varphi\Big) \right] \!,
    \end{align*}
where $\pi_{12}^{\varepsilon} = \big\{(\vec{x},\vec{y}) \in \Omega_{a}\,\big|\; \abs{\vec{x}} > \varepsilon \big\}$, while $d\Sigma$ and $\vec{\nu}$ denote respectively the natural surface measure and the outer unit normal in the boundary integral.

\n
The remaining integral over $\Omega_{a}\!\setminus\! \pi_{12}^{\varepsilon}$ is zero for all $\varepsilon > 0$, given that $\psi_{n}$ solves the eigenvalue equation associated to $\mu_{n}$ outside of the coincidence hyperplane $\pi_{12}$. We shall now examine the integral over $\partial \pi_{12}^{\varepsilon}$. On one hand, using the explicit expression~\eqref{psi12n} for $\psi_n$ together with the identity~\eqref{Kint}, we get
    \begin{align*}
        & \left|\int_{\partial \pi_{12}^{\varepsilon}} \hspace{-0.5cm} d\Sigma\; \overline{\psi_n}\; \partial_{\vec{\nu}} \varphi\right| 
            \leqslant \varepsilon^2\! \int_{B_{a}^{c}}\!\!\!\! d\vec{y} \int_{\mathbb{S}^2}\!\!\!\! d\vec{\omega} \left|\psi_n\big(\varepsilon,|\vec{y}|\big)\right|\, \big|\nabla \varphi(\varepsilon \vec{\omega},\vec{y})\big| 
            \leqslant 16\, \pi^2 \varepsilon^2\,\|\varphi\|_{C^{1}}\! \int_{a}^{\infty}\!\! d\rho\;\rho^2\, \big|\psi_n(\varepsilon,\rho) \big| \\
        & \leqslant 4\pi\,C_n\, \varepsilon\,\|\varphi\|_{C^{1}}\! \int_{a}^{\infty}\!\! d\rho\;\rho\, \frac{\sinh \!\left( s_0 \arctan \frac{\rho}{\varepsilon} \right)}{\sinh \left( s_0\, \tfrac{\pi}{2}\right)} \int_{0}^{\infty}\!\! dz\; e^{- \, \tfrac{t_n}{a} \sqrt{\varepsilon^2 + \rho^2} \cosh z} \\
        & \leqslant 4\pi\,C_n\, \varepsilon\,\|\varphi\|_{C^{1}} \left( \int_{a}^{\infty}\!\! d\rho\;\rho\,e^{- \, \tfrac{t_n}{2a} \rho} \right) \left(\int_{0}^{\infty}\!\!\! dz\; e^{- \, \tfrac{t_n}{2} \cosh z} \right)
            \,\xrightarrow{\varepsilon \to 0^{+}}\; 0\,.
    \end{align*}
On the other hand, by an elementary telescopic argument we obtain
    \begin{align}
        & \int_{\partial \pi_{12}^{\varepsilon}} \hspace{-0.5cm} d\Sigma\; \overline{\partial_{\vec{\nu}} \psi_n}\, \varphi 
        = - \,\varepsilon^2 \int_{B_{a}^{c}}\!\!\! d\vec{y} \int_{\mathbb{S}^2}\!\!\! d\vec{\omega}\, \big[\,\partial_{r} \psi_n(r,\rho)\,\big]_{r \,=\, \varepsilon,\;\rho \,=\, |\vec{y}|}\; \varphi(\varepsilon \vec{\omega},\vec{y}) \nonumber \\
        & = - \int_{B_{a}^{c}}\!\!\! d\vec{y} \int_{\mathbb{S}^2}\!\!\! d\vec{\omega} \left[\,r^2\, \partial_{r} \psi_n(r,\rho) + {\xi_n(\rho) \over 4\pi}\,\right]_{r \,=\, \varepsilon,\; \rho\,=\,|\vec{y}|}\, \varphi(\varepsilon \vec{\omega},\vec{y}) \nonumber \\
        & \hspace{3cm} + \int_{B_{a}^{c}}\!\!\! d\vec{y} \int_{\mathbb{S}^2}\!\!\! d\vec{\omega}\; {\xi_{n}(\vec{y}) \over 4\pi}\, \big[ \varphi(\varepsilon\vec{\omega},\vec{y}) - \varphi(\vec{0},\vec{y})\big]
            + \int_{B_{a}^{c}}\!\!\! d\vec{y}\; \xi_n(\vec{y})\, \varphi(\vec{0},\vec{y})\,.
        \label{proofzero}
    \end{align}
By direct computations and simple estimates,\footnote{In particular, we point out that
\begin{align*}
\left|K_{is_0} \!\left( \tfrac{t_n}{a} \sqrt{r^2 + \rho^2} \right) - K_{is_0} \!\left( \tfrac{t_n}{a}\, \rho \right)\right|
\leqslant \frac{t_n}{a} \int_{0}^{r}\! dt\; {t \over \sqrt{t^2 + \rho^2}} \left| K'_{is_0} \!\left( \tfrac{t_n}{a} \sqrt{t^2 + \rho^2} \right) \right|
\leqslant \frac{t_n}{a\,\rho} \int_{0}^{r}\! dt\; t \int_{0}^{\infty}\!\! dz\; \cosh z\, e^{- \, \tfrac{t_n}{a} \sqrt{t^2 + \rho^2} \cosh z}    \,,
\end{align*}
and further notice that
$$ \sup_{z > 0} \left[ {1 \over z} \left(1 - \frac{\sinh\!\left( s_0 \arctan \tfrac{1}{z}\right)}{\sinh\! \left( s_0\, \tfrac{\pi}{2}\right)} \right) \right]
= \lim_{z \to 0^{+}}\! \left[ {1 \over z} \left(1 - \frac{\sinh\!\left( s_0 \arctan \tfrac{1}{z}\right)}{\sinh\! \left( s_0\, \tfrac{\pi}{2}\right)} \right) \right]
= {s_0\,\cosh\!\left(s_0 \tfrac{\pi}{2} \right) \over \sinh\!\left(s_0 \tfrac{\pi}{2} \right)}\,. $$} 
for all $r > 0$ and $\rho > a$ we deduce
    \begin{align*}
    & \left|r^2\, \partial_{r} \psi_n(r,\rho) + {\xi_n(\rho) \over 4\pi}\right|
    = {C_n \over 4\pi \rho} \bigg| K_{is_0} \!\left( \tfrac{t_n}{a}\, \rho \right) -\, \frac{\sinh \!\left( s_0 \arctan \frac{\rho}{r} \right)}{\sinh\! \left( s_0\, \tfrac{\pi}{2}\right)} \, K_{is_0} \!\left( \tfrac{t_n}{a} \sqrt{r^2 + \rho^2} \right) \\
    & \qquad - {s_0\,r\,\rho \over r^2 + \rho^2}\,\frac{\cosh \!\left( s_0 \arctan \frac{\rho}{r} \right)}{\sinh \!\left( s_0 \tfrac{\pi}{2}\right)} \, K_{is_0} \!\left( \tfrac{t_n}{a} \sqrt{r^2 + \rho^2} \right) 
        + {t_n\,r^2  \over a \sqrt{r^2 + \rho^2}}\,\frac{\sinh \!\left( s_0 \arctan \frac{\rho}{r} \right)}{\sinh \!\left( s_0 \tfrac{\pi}{2}\right)}\, K'_{is_0} \!\left( \tfrac{t_n}{a} \sqrt{r^2 + \rho^2} \right) \bigg| \\
    & \leqslant {C_n r \over 4\pi \rho^2} \Bigg[ {\rho \over r} \left|\frac{\sinh \!\left( s_0 \arctan \frac{\rho}{r} \right)}{\sinh\! \left( s_0\, \tfrac{\pi}{2}\right)} - 1 \right| \int_{0}^{\infty}\!\! dz\; e^{- \, \tfrac{t_n}{a} \sqrt{r^2 + \rho^2} \cosh z}
        + {t_n \over a r} \int_{0}^{r}\! dt\; t \int_{0}^{\infty}\!\!\! dz\; \cosh z\; e^{- \, \tfrac{t_n}{a} \sqrt{t^2 + \rho^2} \cosh z}   \\
    & \hspace{4cm} + {s_0\,\cosh \!\left( s_0 \tfrac{\pi}{2} \right) \over \sinh \!\left( s_0 \tfrac{\pi}{2}\right)}\, \int_{0}^{\infty}\!\! dz\;e^{- \, \tfrac{t_n}{a} \sqrt{r^2 + \rho^2} \cosh z}
        + {t_n\,r  \over a} \int_{0}^{\infty}\!\! dz\; \cosh z \,e^{- \, \tfrac{t_n}{a} \sqrt{r^2 + \rho^2} \cosh z} \Bigg] \\
    & \leqslant {C_n r \over 4\pi \rho^2} \left[ {2 s_0\,\cosh\!\left(s_0 \tfrac{\pi}{2} \right) \over \sinh\!\left(s_0 \tfrac{\pi}{2} \right)} \int_{0}^{\infty}\!\! dz\;  e^{- \, \tfrac{t_n}{2} \cosh z} + {3t_n r \over 2 a} \int_{0}^{\infty}\!\!\! dz\; \cosh z\; e^{- \, \tfrac{t_n}{2} \cosh z} \right] e^{- \, \tfrac{t_n}{2a} \rho}\,,
    \end{align*}
which in turn implies
    \begin{align*}
        & \left| \int_{B_{a}^{c}}\!\!\! d\vec{y} \int_{\mathbb{S}^2}\!\!\! d\vec{\omega} \left[\,r^2\, \partial_{r} \psi_n(r,\rho) + {\xi_n(\rho) \over 4\pi}\,\right]_{r \,=\, \varepsilon,\; \rho\,=\,|\vec{y}|}\, \varphi(\varepsilon \vec{\omega},\vec{y}) \right| \\
        & \leqslant 4 \pi C_n\, \varepsilon\left({2 s_0\,\cosh\!\left(s_0 \tfrac{\pi}{2} \right) \over \sinh\!\left(s_0 \tfrac{\pi}{2} \right)} + {3t_n \varepsilon \over 2 a} \right) \|\varphi\|_{C^{0}} \left( \int_{0}^{\infty}\!\! dz\; \cosh z\; e^{- \, \tfrac{t_n}{2} \cosh z}\right) \left(\int_{a}^{\infty}\!\!\! d\rho\;e^{- \, \tfrac{t_n}{2a} \rho} \right)
        \xrightarrow{\varepsilon \to 0^{+}} 0\,.
    \end{align*}
Similar arguments yield
    \begin{align*}
        & \left| \int_{B_{a}^{c}}\!\!\! d\vec{y} \int_{\mathbb{S}^2}\!\!\! d\vec{\omega}\; {\xi_{n}(\vec{y}) \over 4\pi}\, \big[ \varphi(\varepsilon\vec{\omega},\vec{y}) - \varphi(\vec{0},\vec{y})\big] \right|
            \leqslant 4\pi\,C_n\; \varepsilon\, \|\varphi\|_{C^{1}}\! \int_{a}^{\infty}\!\!\! d\rho\;\rho  \int_{0}^{\infty}\!\!\! dz\; e^{- \, \tfrac{t_n}{a}\, \rho \cosh z} \\
        & \leqslant 4\pi\,C_n\; \varepsilon\, \|\varphi\|_{C^{1}}\! \left(\int_{a}^{\infty}\!\!\! d\rho\;\rho\; e^{- \, \tfrac{t_n}{2a}\, \rho} \right) \left( \int_{0}^{\infty}\!\!\! dz\; e^{- \, \tfrac{t_n}{2} \cosh z} \right) 
            \xrightarrow{\varepsilon \to 0^{+}}\, 0\,.
    \end{align*}
Finally, we remark that the last term in~\eqref{proofzero} coincides with $\langle \xi_{n}\,|\, \tau \varphi \rangle$ for any smooth $\varphi$.

The above results prove~\eqref{eq:thesis} for all $\varphi \in \mathcal{D}(\Omega_a)$. Now the thesis follows by plain density arguments. In fact, for any given $\varphi \in \mbox{dom}\big(H_{D}\big) = H^1_{0}(\Omega_a) \cap H^2(\Omega_a)$ there exists an approximating sequence $\{\varphi_{j}\}_{j \in \mathbb{N}} \subset \mathcal{D}(\Omega_a)$, converging to $\varphi$ in the natural topology on $\mbox{dom}\big(H_{D}\big)$ induced by the graph norm, such that
    \begin{align*}
    & \left| \big\langle \psi_n \big| (H_D + \mu_{n}) \varphi \big\rangle - \big\langle \xi_n \big| \tau \varphi \big\rangle \right| \\
    & \leqslant \left| \big\langle \psi_n \big| (H_D + \mu_{n}) ( \varphi - \varphi_j ) \big\rangle \right| 
        + \left|\big\langle \psi_n \big| (H_D + \mu_{n}) \varphi_{j} \big\rangle - \langle \xi_n \big| \tau \varphi_j \rangle \right| 
        + \left| \langle \xi_n \big| \tau ( \varphi_j - \varphi ) \rangle \right| \\
    & \leqslant \|\psi_n\|_{L^2}\,\big( \|H_D ( \varphi - \varphi_j )\|_{L^2} +|\mu_{n}|\, \|\varphi - \varphi_j\|_{L^2} \big)
        + \|\xi_{n}\|_{L^2}\, \|\varphi_j - \varphi \|_{H^2}
        \,\xrightarrow{j \to +\infty}\, 0\,.
    \end{align*}    
\end{proof}

We are now ready to prove Theorem \ref{thm:theig}.

\begin{proof}[Proof of Theorem \ref{thm:theig}]
Recall that, for all $n \in \mathbb{N}$, the function $\Psi_{n}$ given by~\eqref{Psin}, \eqref{psi12n} is a formal eigenfunction of $H_{D,0}$ by construction. Lemma \ref{eq:FaddeevComponent} further ensures that $\Psi_n = G^{\mu_n} \xi_n$, where $\mu_n$ is fixed according to~\eqref{mun} and we are employing the slightly abusive notation~\eqref{eq:Gz}. Since $\xi_n \in \mbox{dom}\big(\Gamma_{0}^{\lambda}\big)$, this suffices to infer that $\Psi_n \in \mbox{dom}\big(H_{D,0}\big)$. Moreover, making reference to Remark \ref{rem:eig}, let us stress that the boundary condition for $\Psi_n = 0 + G^{\mu_n} \xi_n$ encoded in $\mbox{dom}\big(H_{D,0}\big)$ reduces to $\Gamma_{0}^{\lambda} \xi_{n} = 0$.
\end{proof}


\newpage

\appendix
\section{On the integral kernel for the Dirichlet resolvent}\label{app:g0}

In this appendix we collect some results regarding the integral kernel $R_{D}^{\lambda}\big(\vec{X},\vec{X}'\big)$ associated to the Dirichlet resolvent $R^\lambda_D := (H_D + \lambda)^{-1}$ ($\lambda > 0$). We recall that throughout the paper we refer to the decomposition~\eqref{eq:RDG0gz}, namely,
    \begin{equation*}
        R_{D}^{\lambda}\big(\vec{X},\vec{X}'\big) = R_{0}^{\lambda}\big(\vec{X},\vec{X}'\big) + g^{\lambda}\big(\vec{X},\vec{X}'\big)   \,,
    \end{equation*}
where $R_{0}^{\lambda}\big(\vec{X},\vec{X}'\big)$ is the resolvent kernel associated to the free Laplacian in $\mathbb{R}^6$ and, for fixed $\vec{X}' \in \Omega_a$, $g^{\lambda}\big(\vec{X},\vec{X}'\big)$ is the solution of the elliptic problem~\eqref{eq:gzPDE}, \emph{i.e.},
    \begin{equation*}
        \left\{\begin{array}{ll} \!
            \displaystyle{(- \Delta_{\vec{X}} +\lambda) g^\lambda\big(\vec{X},\vec{X}'\big) = 0} & \displaystyle{\mbox{for\; $\vec{X} \in \Omega_{a}$}}\,, \vspace{0.15cm}\\
            \displaystyle{g^\lambda\big(\vec{X},\vec{X}'\big) = -\,R_{0}^{\lambda}\big(\vec{X},\vec{X}'\big)} & \displaystyle{\mbox{for\; $\vec{X} \in \partial\Omega_{a}$}}\,, \\
            \displaystyle{g^\lambda\big(\vec{X},\vec{X}'\big) \to 0} & 
            \displaystyle{\mbox{for\; $\abs{\vec{X}} \to +\infty$}}\,.
        \end{array}\right.      
    \end{equation*}
    
We first remark that by elementary spectral arguments it follows that
    \begin{equation}\label{RDR0sym}
        R_{D}^{\lambda}\big(\vec{X},\vec{X}'\big) = R_{D}^{\lambda}\big(\vec{X}'\!,\vec{X}\big)\,, \qquad
        R_{0}^{\lambda}\big(\vec{X},\vec{X}'\big) = R_{0}^{\lambda}\big(\vec{X}'\!,\vec{X}\big)\,,
    \end{equation}
which entails, in turn,
    \begin{equation*}
        g^{\lambda}\big(\vec{X},\vec{X}'\big) = g^{\lambda}\big(\vec{X}'\!,\vec{X}\big)\,.
    \end{equation*}

Regarding $R_{0}^{\lambda}\big(\vec{X},\vec{X}'\big)$, a well known computation yields
    \begin{equation}\label{eq:Gz0}
        R_{0}^{\lambda}\big(\vec{X},\vec{X}'\big) 
        = {1 \over (2\pi)^6}\int_{\mathbb{R}^6}\hspace{-0.1cm} d\vec{K}\; {e^{i \vec{K} \cdot (\vec{X} - \vec{X}')} \over |\vec{K}|^2 +\lambda}
        = {\lambda \over (2\pi)^3}\,{K_{2}\big(\sqrt{\lambda}\, |\vec{X} - \vec{X}'|\big) \over |\vec{X} - \vec{X}'|^2}\,, 
    \end{equation}
where $K_2$ is the modified Bessel function of second kind, \emph{a.k.a.} Macdonald function. Then, using a noteworthy summation theorem for Bessel functions \cite[p. 940, Eq. 8.532 1]{GR}, for $ \abs{\vec{X}} \neq |\vec{X}'|$ we obtain\footnote{The identity~\eqref{eq:Gz0ser} holds, in principle, only for $\abs{\vec{X}} <|\vec{X}'|$. Yet, it can be readily extended to the whole set $\abs{\vec{X}} \neq |\vec{X}'|$ using the basic symmetry relation for $R_{0}^{\lambda}\big(\vec{X},\vec{X}'\big)$ in~\eqref{RDR0sym}.}
    \begin{equation}\label{eq:Gz0ser}
        R_{0}^{\lambda}\big(\vec{X},\vec{X}'\big) = {1 \over 2\pi^3} \sum_{\ell = 0}^{\infty} (\ell + 2)\,C_{\ell}^{2}\!\left({\vec{X} \cdot \vec{X}' \over \abs{\vec{X}}\,|\vec{X}'|}\right) {I_{\ell + 2}\big(\sqrt{\lambda}\, \abs{\vec{X}}\big) \over \abs{\vec{X}}^2}\,{K_{\ell + 2}\big(\sqrt{\lambda}\, |\vec{X}'|\big) \over |\vec{X}'|^2}\;,
    \end{equation}
    where $C_{\ell}^{2}$ are the Gegenbauer (ultraspherical) polynomials defined by the identity \cite[\S 8.930]{GR}
        \begin{equation}\label{eq:Gegen}
            {1 \over (1 - 2 s u + u^2)^2} = \sum_{\ell = 0}^{\infty} C_{\ell}^{2}(s)\,u^\ell\,, \qquad \mbox{for\; $s \in [-1,1]$,\, $u \in (-1,1)$\,.}
        \end{equation}  

In the last part of this appendix we derive a series representation for the remainder function $g^{\lambda}(\vec{X},\vec{X}')$. Without loss of generality, we shall henceforth assume $\vec{X}' = (\vec{x}',\vec{y}')$ to lie on the 6\textsuperscript{th} axis, namely $\vec{X}' = y'_3\,\vec{e}_6$ with $\vec{e}_6 = (0,0,0,0,0,1)$. To simplify the notation, in the sequel we drop the dependence on $\vec{X}'$ and put
\begin{equation}
    g^\lambda(\vec{X}) \equiv g^\lambda(\vec{X},\vec{X}')\,.
\end{equation}
To proceed, we refer to the classical representation of the Laplace operator in hyper-spherical coordinates~\cite{avery}. More precisely, let us introduce the set of coordinates $(r,\vec{\omega}) \in \mathbb{R}_{+} \times \mathbb{S}^5$ and recall that
\begin{equation}\label{eq:Lap6D}
    - \Delta_{\vec{X}} = - \,{1 \over r^5}\, \partial_r \big(r^5\,\partial_r \cdot\big) - \Delta_{\mspace{1.5mu}\mathbb{S}^5}\,,
\end{equation}
where $\Delta_{\mspace{1.5mu}\mathbb{S}^5}$ is the Laplace-Beltrami operator on $\mathbb{S}^5$. The latter operator is essentially self-adjoint in $L^2(\mathbb{S}^5)$ and has pure point spectrum consisting of degenerate eigenvalues $\sigma(-\Delta_{\mspace{1.5mu}\mathbb{S}^5}) = \{\ell(\ell + 4)\,|\,\ell  = 0,1,2,\,\dots\, \}$. Correspondingly, a complete orthonormal set of eigenfunctions is given by the hyper-spherical harmonics $\mathcal{Y}_{\ell,\vec{m}}$, with $\ell \in \N_0$ and $\vec{m} = (m_0, \ldots , m_4) \in \Z^5$ such that $\ell = m_0 \geqslant m_1 \geqslant m_2 \geqslant m_3 \geqslant |m_4| \geqslant 0$.
Taking this into account, we make the following ansatz for the generic solution of the differential equation $(-\Delta_{\vec{X}} +\lambda)g^\lambda = 0$:
\begin{equation*}
    g^\lambda(r,\vec{\omega}) \equiv g^\lambda\big(\vec{X}(r,\vec{\omega})\big) = \sum_{\ell,\vec{m}} \mathcal{R}^{\lambda}_{\ell,\vec{m}}(r)\,\mathcal{Y}_{\ell,\vec{m}}(\vec{\omega})\,.
\end{equation*}
Using the basic identity~\eqref{eq:Lap6D}, we obtain an ODE for $\mathcal{R}^{\lambda}_{\ell,\vec{m}}$. The solutions are of the form
\begin{equation*}
    \mathcal{R}^{\lambda}_{\ell,\vec{m}}(r) = \alpha^{(I)}_{\ell,\vec{m}}\,{I_{\ell + 2}\big(\sqrt{\lambda}\,r\big) \over r^{2}} + \alpha^{(K)}_{\ell,\vec{m}}\,{K_{\ell + 2}\big(\sqrt{\lambda}\,r\big) \over r^{2}}  \qquad \mbox{for some constants $\alpha^{(I)}_{\ell,\vec{m}}, \alpha^{(K)}_{\ell,\vec{m}} \in \mathbb{R}$}\,.
\end{equation*}
Considering the asymptotic behavior of the Bessel functions $I_{\nu},K_{\nu}$ with large arguments \cite[\S 10.30(ii)]{NIST}, it is necessary to fix $\alpha^{(I)}_{\ell,\vec{m}} = 0$ to fulfill the condition $g^\lambda \to 0$ for $r \to +\infty$. Furthermore, let us point out that the solution $g^{\lambda}$ has to be invariant under rotations around the fixed vector $\vec{X}'$. Keeping in mind that we chose $\vec{X}'$ to lie on the 6\textsuperscript{th} axis, this means that we have to fix $\alpha^{(K)}_{\ell,\vec{m}} = 0$ for all $\vec{m} \neq (\ell,0,0,0,0)$ ($\ell \in \N_0$). 
The previous arguments, together with a sum rule for the hyper-spherical harmonics $\mathcal{Y}_{\ell,\vec{m}}$ \cite[p. 1372, Eq. 66]{avery}, entail
\begin{equation*}
    g^\lambda(r,\vec{\omega}) = \sum_{\ell \,=\, 0}^{+ \infty} \alpha_{\ell}\,C_\ell^2(\vec{\omega} \cdot \vec{e}_6)\, {K_{\ell + 2}\big(\sqrt{\lambda}\,r\big) \over r^{2}}\,, 
\end{equation*}
or, equivalently,
\begin{equation}\label{eq:gzserAPP}
    g^\lambda(\vec{X}) = \sum_{\ell \,=\, 0}^{+ \infty} \alpha_{\ell}\,C_\ell^2\left({\vec{X} \cdot \vec{X}' \over \abs{\vec{X}}\,|\vec{X}'|}\right){K_{\ell + 2}\big(\sqrt{\lambda}\,\abs{\vec{X}}\big) \over \abs{\vec{X}}^{2}}\, .
\end{equation}
Here, $C_{\ell}^2$ are the Gegenbauer polynomials defined by~\eqref{eq:Gegen} and $(\alpha_{\ell})_{\ell \,=\,0,1,2,\ldots} \subset \mathbb{R}$ are suitable coefficients. We now fix these coefficients so as to fulfill the non-homogeneous Dirichlet boundary condition in the second line of~\eqref{eq:gzPDE}. In view of the identity~\eqref{eq:Gz0ser}, the said boundary condition for $\abs{\vec{X}} = a$ becomes
\begin{equation*}
    \sum_{\ell \,=\, 0}^{+ \infty} \alpha_{\ell}\,C_\ell^2\left({\vec{X} \cdot \vec{X}' \over \abs{\vec{X}}\,|\vec{X}'|}\right) {K_{\ell + 2}\big(\sqrt{\lambda}\,a\big) \over a^{2}} 
    = -\,{1 \over 2\pi^3} \sum_{\ell \,=\, 0}^{\infty} (\ell + 2)\,C_{\ell}^{2}\!\left({\vec{X} \cdot \vec{X}' \over \abs{\vec{X}}\,|\vec{X}'|}\right) {I_{\ell + 2}\big(\sqrt{\lambda}\, a\big) \over a^2}\,{K_{\ell + 2}\big(\sqrt{\lambda}\, |\vec{X}'|\big) \over |\vec{X}'|^2}\,.
\end{equation*}
Upon varying $\vec{X} \in \partial\Omega_a$, this implies
\begin{equation*}
     \alpha_{\ell} = -\,{1 \over 2\pi^3}\, (\ell + 2)\, {I_{\ell + 2}\big(\sqrt{\lambda}\, a\big) \over K_{\ell + 2}\big(\sqrt{\lambda}\,a\big)}\,{K_{\ell + 2}\big(\sqrt{\lambda}\, |\vec{X}'|\big) \over |\vec{X}'|^2}\,,
\end{equation*}
which, together with Eq.~\eqref{eq:gzserAPP}, ultimately yields
    \begin{equation}\label{eq:gzser}
        g^{\lambda}\big(\vec{X}\,;\vec{X}'\big) = -\,{1 \over 2\pi^3} \sum_{\ell \,=\, 0}^{\infty}\, (\ell + 2)\,C_{\ell}^{2}\!\left({\vec{X} \cdot \vec{X}' \over \abs{\vec{X}}\,|\vec{X}'|}\right) {I_{\ell + 2}\big(\sqrt{\lambda}\, a\big) \over K_{\ell + 2}\big( \sqrt{\lambda}\, a\big)}\,{K_{\ell + 2}\big(\sqrt{\lambda}\, \abs{\vec{X}}\big) \over \abs{\vec{X}}^2}\,{K_{\ell + 2}\big(\sqrt{\lambda}\, |\vec{X}'|\big) \over |\vec{X}'|^2}\,.
    \end{equation}

    Let us mention the following asymptotic expansions, for any fixed $s \in [-1,1]$ and $t > 0$ \cite[p. 256, Eqs. 10.41.1-2 and p.450, Eq. 18.14.4, together with p.136, Eq. 5.2.5 and p.140, Eq. 5.11.3]{NIST}:
    \begin{equation*}
        I_{\nu}(t) \lesssim {1 \over \sqrt{2\pi \nu}} \left({e\,t \over 2\nu}\right)^{\!\nu}, \qquad
        K_{\nu}(t) \lesssim \sqrt{{\pi \over 2 \nu}} \left({e\,t \over 2\nu}\right)^{\!-\nu}, \qquad
        C^{2}_{\nu}(s) \lesssim \nu^3, \qquad\quad
        \mbox{for\, $\nu \to \infty$}\,.
    \end{equation*}
    Taking these into account it is easy to see that, for any fixed $\vec{X},\vec{X}' \in \Omega_a$, the series in Eq.~\eqref{eq:gzser} behaves as
    \begin{equation*}
        \sum_{\ell \,=\, 1}^{\infty} \ell^3 \left({a \over \abs{\vec{X}}}\right)^{\!\!\ell} \left({a \over |\vec{X}'|}\right)^{\!\!\ell} ,  
    \end{equation*}
    which suffices to infer that~\eqref{eq:gzser} makes sense as a pointwise convergent series. 


\section{Derivation of \texorpdfstring{$\mu_n$}{the eigenvalues} and \texorpdfstring{$\psi_n$}{the eigenvectors}}\label{bvp}

\n
Let $\psi$ be the solution of the boundary value problem~\eqref{eqra1}, \eqref{dirbc}, \eqref{bc1}. If we define the function $\zeta(r,\rho) = r \,\rho \,\psi(r,\rho)$, then the corresponding problem for $\zeta$ reads
\begin{align}
&- \frac{\partial^2 \zeta}{\partial r^2} - \frac{\partial^2 \zeta}{\partial \rho^2} +\mu\, \zeta=0  \hspace{2.7cm} \text{in\; $D_{a}$}\,, \\
&\zeta (r,\rho)=0 \quad \text{for $r^2 + \rho^2 =a^2$}\,, \hspace{1.6cm} \zeta(r,0)=0 \quad\text{for\; $r \geqslant a$}\,, \\
&\frac{\partial \zeta}{\partial r} (0,\rho) +\frac{8}{\sqrt{3}\,\rho} \, \zeta\left(\tfrac{\sqrt{3}}{2}\, \rho, \tfrac{1}{2}\, \rho\right) = 0 \hspace{1cm} \text{for\; $\rho \geqslant a$}\,.
\end{align}
Using polar coordinates $(r,\omega) \in \mathbb{R}_{+} \times (0,\pi/2)$, the above problem can be solved by separation of variables. Indeed, defining $\eta(R,\omega)= \zeta (R \sin \omega, R \cos \omega)$, one finds
\begin{align}
&-\frac{\partial^2 \eta}{\partial R^2} - \frac{1}{R} \frac{\partial \eta}{\partial R} - \frac{1}{R^2} \frac{\partial^2 \eta}{\partial \omega^2} + \mu\, \eta =0 \hspace{1.5cm} \text{for\, $R>a$,\, $\omega \in (0, \tfrac{\pi}{2})$}\,, \\
& \eta \left(R, \tfrac{\pi}{2}\right) =0 \quad \text{for\, $R\geq a$}\,, \hspace{3.2cm} \eta(a,\omega)=0 \quad \text{for\, $\omega \in (0,\tfrac{\pi}{2})$}\,, \\
&\frac{\partial \eta}{\partial \omega} (R,0)+ \frac{8}{\sqrt{3}} \, \eta\left(R, \tfrac{\pi}{3}\right) =0 \hspace{2.75cm} \text{for\, $R > a$}\,.
\end{align}
Let us now look for solutions in the product form $\eta (R,\omega)= f(R)\,g(\omega)$. Then
\begin{align}\label{prg}
g'' - \nu g =0\,, \qquad g(\tfrac{\pi}{2})=0\,, \qquad g'(0) + \frac{8}{\sqrt{3}}\, g\left(\tfrac{\pi}{3}\right)=0\,,
\end{align}
and
\begin{align}\label{eqf}
f'' + \frac{1}{R}\, f'+ \left( \frac{\nu}{R^2} - \mu \right) f=0\,, \qquad f(a)=0\,,
\end{align}
where $\nu$ is a real separation constant. We are interested in the case $\nu>0$ and it turns out that in this case the only solution (apart from a multiplicative factor) of problem~\eqref{prg} is 
\begin{align}\label{g0}
g_0(\omega)= \sinh \left[ s_0 \left(\omega - \tfrac{\pi}{2}\right)\right]\,,
\end{align}
where $\nu= s_0^2$ and $s_0>0$ is the only positive solution of Eq.~\eqref{eqs0}. 
Next, one solves problem~\eqref{eqf} with $\nu=s_0^2$. The only solution (apart from a multiplicative factor) of the differential equation going to zero for $R\to \infty$ is the modified Bessel function of imaginary order $K_{is_0}\big(\sqrt{\mu} \,R\big)$. It remains to impose the Dirichlet boundary condition $K_{is_0}\big(\sqrt{\mu} \,a\big)=0$, which dictates the choice $\mu=\mu_n$ as in~\eqref{mun}. 
Accordingly, the problem~\eqref{eqf} for $\mu=\mu_n$ has a solution going to zero for $R\to \infty$ given by
\begin{align}\label{fnR}
f_n(R)= K_{is_0}\!\left( \frac{t_n}{a} \,R\right) .
\end{align}
By~\eqref{g0} and~\eqref{fnR} we reconstruct the solution~\eqref{psi12n} for the boundary value problem ~\eqref{eqra1}, \eqref{dirbc}, \eqref{bc1}. 
\vspace{0.5cm}\\

\n
\textbf{Acknowledgments}\\
We warmly thank Andrea Posilicano for helpful discussions on the subject of this work.

\end{document}